\documentclass[letter,11pt,dvipsnames]{article}
\newcommand{\ignore}[1]{}

\usepackage{verbatim} 
\usepackage[x11names, rgb]{xcolor}
\usepackage[utf8]{inputenc}
\usepackage{tikz}
\usetikzlibrary{snakes,arrows,shapes}
\usepackage{fullpage}
\usepackage[colorlinks=false, pagebackref]{hyperref}  
\usepackage{microtype}
\usepackage{textcomp}
\usepackage{amsmath, amsfonts, amssymb, amsthm}
\usepackage{mathtools, mathdots}
\usepackage{mhsetup}
\usepackage{algorithm, algpseudocode, algorithmicx, float}
\usepackage{enumerate}
\usepackage{comment}

\newtheorem{theorem}{Theorem}

\newtheorem{corollary}[theorem]{Corollary}
\newtheorem{lemma}[theorem]{Lemma}
\newtheorem{definition}[theorem]{Definition}

\newtheorem{claim}[theorem]{Claim}

\newcommand{\claimproof}[2]%
{\noindent{\em Proof of Claim \ref{#1}.}
#2\hspace*{\fill}$\Box$~~~~\vspace{3.5mm} }

\newcommand{\G}{\mathbb{G}}
\newcommand{\F}{\mathbb{F}}
\newcommand{\N}{\mathbb{N}}
\newcommand{\R}{\mathbb{R}}
\newcommand{\Q}{\mathbb{Q}}
\newcommand{\Z}{\mathbb{Z}}

\newcommand{\calZ}{\mathcal{Z}}

\newcommand{\poly}{\text{poly}}

\newcommand{\len}{\text{len}}
\newcommand{\des}{\text{des}}

\makeatletter
\newenvironment{breakablealgorithm}
  {
   \begin{center}
     \refstepcounter{algorithm}
     \hrule height.8pt depth0pt \kern2pt
     \renewcommand{\caption}[2][\relax]{
       {\raggedright\textbf{\ALG@name~\thealgorithm} ##2\par}%
       \ifx\relax##1\relax 
         \addcontentsline{loa}{algorithm}{\protect\numberline{\thealgorithm}##2}%
       \else 
         \addcontentsline{loa}{algorithm}{\protect\numberline{\thealgorithm}##1}%
       \fi
       \kern2pt\hrule\kern2pt
     }
  }{
     \kern2pt\hrule\relax
   \end{center}
  }
\makeatother

\makeatletter
\newcommand{\algmargin}{\the\ALG@thistlm}
\makeatother
\newlength{\whilewidth}
\algdef{SE}[parWHILE]{parWhile}{EndparWhile}[1]
  {\parbox[t]{\dimexpr\linewidth-\algmargin}{%
     \hangindent\whilewidth\strut\algorithmicwhile\ #1\ \algorithmicdo\strut}}{\algorithmicend\ \algorithmicwhile}%
\algnewcommand{\parState}[1]{\State%
  \parbox[t]{\dimexpr\linewidth-\algmargin}{\strut #1\strut}}

\begin{document}

\pagenumbering{gobble}


\title{\bf Counting basic-irreducible factors mod $p^k$ in deterministic poly-time and $p$-adic applications}

\author{
Ashish Dwivedi \thanks{CSE, Indian Institute of Technology, Kanpur, \texttt{ashish@cse.iitk.ac.in} }
\and
Rajat Mittal \thanks{CSE, Indian Institute of Technology, Kanpur, \texttt{rmittal@cse.iitk.ac.in} }
\and
Nitin Saxena \thanks{CSE, Indian Institute of Technology, Kanpur, \texttt{nitin@cse.iitk.ac.in} }
}

\date{}
\maketitle

\begin{abstract}
Finding an irreducible factor, of a polynomial $f(x)$ modulo a prime $p$, is not known to be in deterministic polynomial time. Though there is such a classical algorithm that {\em counts} the number of irreducible factors of $f\bmod p$. We can ask the same question modulo prime-powers $p^k$. The irreducible factors of $f\bmod p^k$ blow up exponentially in number; making it hard to describe them. Can we count  those irreducible factors $\bmod~p^k$ that remain irreducible mod $p$? These are called {\em basic-irreducible}. A simple example is in $f=x^2+px \bmod p^2$; it has $p$ many basic-irreducible factors. Also note that, $x^2+p \bmod p^2$ is irreducible but not basic-irreducible!

We give an algorithm to count the number of basic-irreducible factors of $f\bmod p^k$ in deterministic poly($\deg(f),k\log p$)-time. This solves the open questions posed in (Cheng et al, ANTS'18 \& Kopp et al, Math.Comp.'19). In particular, we are counting roots $\bmod\ p^k$; which gives the first deterministic poly-time algorithm to compute Igusa zeta function of $f$. Also, our algorithm efficiently partitions the set of all basic-irreducible factors (possibly exponential) into merely $\deg(f)$-many disjoint sets, using a compact tree data structure and {\em split} ideals. 

\end{abstract}

\vspace{-.30mm}
\noindent
{\bf 2012 ACM CCS concept:} Theory of computation-- Algebraic complexity theory,  Pseudorandomness and derandomization; Computing methodologies-- Algebraic/ Number theory algorithms, Hybrid symbolic-numeric methods; Mathematics of computing-- Combinatoric problems.

\vspace{-.45mm}


\pagenumbering{arabic}

\vspace{-1mm}

\section{Introduction}
\vspace{-1mm}

Factoring a univariate polynomial, over {\em prime} characteristic, is a highly well studied problem. Though efficient factoring has been achieved using randomization, still efficient derandomization is a longstanding problem. A related question of equal importance is root finding, but this is known to be equivalent to factoring in deterministic poly-time. Surprisingly, testing irreducibility, or even counting irreducible factors, is easy in this regime. The main tool here is the magical Frobenius morphism of prime $p$ characteristic rings:  $x\mapsto x^p$.

Though much effort has been put in prime characteristic, few results are known in {\em composite} characteristic $n$ \cite{shamir1993generation}. Even irreducibility testing of a polynomial, with the prime factorization of $n$ given, has no efficient algorithm known. This reduces to {\em prime-power} characteristic $p^k$ \cite{von1998factoring}. Deterministic factoring in such a ring is a much harder question (at least it subsumes deterministic factoring mod $p$). In fact, even randomized algorithms, or practical solutions, are currently elusive \cite{von1996factorization, von1998factoring, klivans1997factoring, salagean, sircana2017factorization, dwivedi2019efficiently}. The main obstruction is non-unique factorization.

Being a non-unique factorization domain, there could be exponential number of roots, or irreducible factors, modulo prime-powers \cite{von1996factorization}. So one could ask a related question about counting all the irreducible factors (respectively roots) modulo prime-powers. Efficiently solving this counting problem will give us an efficient irreducibility testing criteria, which is the first question one wants to try. Recall that prime characteristic allows such an efficient method.

Motivated by this, we ask--- Could we describe all factors which remain irreducible mod $p$? Such factors are called {\em basic}-irreducible in the literature. This is much more than counting roots mod $p^k$ (as, $f(\alpha)=0$ iff $x-\alpha$ is a basic-irreducible factor of $f$). These roots, besides being naturally interesting, have various applications in--- factoring \cite{chistov1987efficient, chistov1994algorithm, cantor2000factoring}, coding theory \cite{berthomieu2013polynomial, salagean}, elliptic curve cryptography \cite{lauder2004counting}, arithmetic algebraic-geometry \cite{zuniga2003computing, denef2001newton, denef1991report, igusa1974complex}. Towards this we design a machinery, yielding the following result:

{\em Given a degree $d$ integral polynomial $f(x)$ and a prime-power $p^k$, we partition the set of all basic-irreducible factors of $f\bmod p^k$ into at most $d$ (compactly provided) subsets in deterministic $\poly(d, k\log p)$-time; in the same time we count the number of factors in each of these subsets.

Also, we can compactly partition (and count) the roots of $f\bmod p^k$ in deterministic poly-time. }

\smallskip
This efficient partitioning of (possibly exponentially many) roots into merely $d$ subsets is reminiscent of the age-old fact: there are at most $\deg(g)$ roots of a polynomial $g(x)$ over a field.
Root sets mod $p^k$ are curious objects; not every subset of $\Z/p^k\Z$ is a root set (except when $k=1$). Their combinatorial properties have been studied extensively \cite{sierpinski1955remarques, chojnacka1956congruences, bhargava1997p, dearden1997roots, maulik2001root}. In this regard, our result is one more step to understand the hidden properties of root-sets mod prime-powers.

\smallskip
Factoring mod $p^k$ has applications in factoring over {\em local} fields \cite{chistov1987efficient, chistov1994algorithm, cantor2000factoring}. Previously, the latter was achieved through randomized factoring mod $p$ \cite{cantor1981new} and going to extensions of $\Q_p$. Directly factoring mod $p^k$, for arbitrary $k$, would imply a new and more natural factoring algorithm over $p$-adic fields. In fact, {\em our method gives the first deterministic poly-time algorithm to count basic-irreducible factors of $f\in \Q_p[x]$}; by picking $k$ such that $p^k\nmid\text{ discriminant}(f)$. This derandomization was not known before, though $\Q_p[x]$ is indeed a unique factorization domain.

\vspace{-2mm}
\subsection{Previously known results}
\vspace{-1mm} 

The questions of root finding and root counting of $f\bmod p^k$ are of classical interest, see \cite{niven2013introduction, apostol2013introduction}.
Using Hensel lifting (Section \ref{appen-prelim}) we know how to `lift' a root, of multiplicity one, of $f$ mod $p$ to a root of $f$ mod $p^k$, in a unique way. But this method breaks down when the root (mod $p$) has multiplicity more than one. \cite[Cor.4]{berthomieu2013polynomial} was the first work to give an efficient randomized algorithm to count, as well as find, all the roots of $f$ mod $p^k$. In this line of progress, very recently \cite{cheng2017counting} gave a deterministic algorithm to count roots in time {\em exponential} in the parameter $k$. Extending the idea of \cite{cheng2017counting}, \cite{kopp2018randomized} gave another efficient randomized algorithm to count roots of $f\bmod p^k$. Note that finding the roots deterministically seems a difficult problem because it requires efficient deterministic factoring of $f$ mod $p$ (which is a classical open problem). But counting the roots mod $p^k$ deterministically may be an easier first step.

Recently there has been some progress in factoring $f\bmod p^k$ when $k$ is constant. \cite{dwivedi2019efficiently} gave the first efficient randomized algorithm to factor $f\bmod p^k$ for $k\leq 4$. This gives an exponential improvement over the previous best algorithms of \cite{sircana2017factorization, von1998factoring, von1996factorization} mod $p^k$ ($k\leq 4$). In fact, they generalized Hensel lifting method to mod $p^k$, for $k\leq 4$, in the difficult case when $f\bmod p$ is power of an irreducible. The related derandomization questions are all open.

The case of factoring $f\bmod p^k$ when $k$ is ``large''--- larger than the maximum power of $p$ dividing the discriminant of the integral $f$ ---has an efficient randomized algorithm due to \cite{von1998factoring}. They showed, assuming large $k$, that factorization mod $p^k$ is well behaved and corresponds to the unique $p$-adic factorization of $f$ (i.e.~in $\Q_p[x]$). In turn, $p$-adic factoring has known efficient randomized algorithms \cite{chistov1987efficient, chistov1994algorithm, cantor2000factoring}. The derandomization questions are all open.
 
We now give a deterministic method to count all the roots (resp.~basic-irreducible factors) efficiently. In fact, our proof can be seen as a deterministic poly-time reduction of basic-irreducible factor finding mod $p^k$ to root finding mod $p$. 
In particular, it subsumes all the results of \cite{berthomieu2013polynomial}.

\vspace{-2mm}
\subsection{Our results }\label{sec-results}
\vspace{-1mm}


\vspace{-0.5mm}
\begin{theorem}[Root count]
\label{thm1}
Let $p$ be a prime, $k\in\N$ and $f(x)\in\Z[x]$. Then, all the roots of $f \bmod p^k$ can be counted in deterministic poly($\deg f, k\log p$)-time. 
\end{theorem}

This is the first efficient derandomization of the randomized root counting algorithms \cite{berthomieu2013polynomial, kopp2018randomized}, and an exponential improvement over the recent deterministic algorithm of \cite{cheng2017counting}. The challenge arises from the fact that we need to count the possibly exponentially many roots without being able to find them.

\medskip\noindent{\bf Remarks.}
{\bf 1)} 
In the algorithm, the (possibly exponential) root-set of $f \bmod p^k$ gets partitioned into at most $\deg(f)$-many disjoint subsets and we output a compact representation, called {\em split ideal}, for each of these subsets. We do count them, but do not yet know how to find a root deterministically.

{\bf 2)} 
This gives an efficient way to deterministically compute the Igusa zeta function, given an integral univariate $f$ and a prime $p$. This follows from the fact that we just need to compute $N_k(f):= $number of roots of $f$ mod $p^k$, for $k\in[\ell]$ s.t.~$p^\ell\nmid\text{ discriminant}(f)$, to estimate {\em Poincar\'e series} $\sum_{i=0}^{\infty} N_i(f)x^i$ \cite{denef1991report, igusa1974complex}. Interestingly, it converges to a rational function!

{\bf 3)} 
This is the first deterministic poly-time algorithm to count the number of lifts of a {\em repeated} root of $f \bmod p$ to $f \bmod p^k$.

{\bf 4)} 
This gives the first deterministic poly-time algorithm to count the number of $p$-adic integral roots of a given $p$-adic polynomial $f\in \Q_p[x]$. (Count roots mod $p^\ell$ where $p^\ell\nmid\text{ discriminant}(f)$.)

\smallskip
Next, we extend the ideas for counting roots to count all the basic-irreducible factors of $f\bmod p^k$ in deterministic polynomial time. Recall that a {\em basic-irreducible} factor of $f\bmod p^k$ is one that remains irreducible in mod $p$ arithmetic. 

\vspace{-0.7mm}
\begin{theorem}[Factor count]\label{thm2}
Let $p$ be a prime, $k\in\N$ and $f(x)\in\Z[x]$. Then, all the basic-irreducible factors of $f \bmod p^k$ can be counted in deterministic poly($\deg f, k\log p$)-time. 
\end{theorem}

We achieve this by extending the idea of counting roots to more general $p$-adic integers. Essentially, we efficiently count all the roots of $f(x)$ in $\mathcal{O}_K/\langle p^k\rangle$, where $\mathcal{O}_K$ is the ring of integers of a $p$-adic {\em unramified} extension $K/\mathbb{Q}_p$ (refer \cite{koblitz1977p} for the standard notation). Currently, there is no fast, practical method known to find/count roots when $K$ is {\em ramified}.

\begin{corollary}
Consider (an unknown) $p$-adic extension $K:= \mathbb{Q}_p[y]/\langle g(y)\rangle$, which is unramified and has degree $\Delta$. Let $f(x) \in \mathbb{Z}[x]$, $p, k, \Delta$ be given as input (in binary).

Then, we can count all the roots of $f$, in $\mathcal{O}_K/\langle p^k\rangle$, in deterministic poly($\deg(f), k\log p, \Delta$)-time.
\end{corollary}
\noindent{\bf Remarks.}
{\bf 1)} 
This gives the first deterministic poly-time algorithm to count the number of (unramified $p$-adic integral) roots of a given $p$-adic polynomial $f\in K[x]$.

{\bf 2)} 
Our method generalizes to efficiently count all the roots of a given polynomial $f(x) \in (\mathbb{F}[t]/\langle h(t)^k\rangle)[x]$ for a given polynomial $h$ (resp.~$f\in\F[[t]][x]$ with power-series coefficients); assuming that $\mathbb{F}$ is a field over which root counting is efficient (eg.~$\mathbb{Q}, \R, \mathbb{F}_p$ and their algebraic extensions).

\vspace{-1mm}
\subsection{Proof techniques} \label{sec-tech}
\vspace{-1mm}

Our implementation involves constructing a {\em list data structure} $\mathcal{L}$ which implicitly partitions the root-set of $f\bmod p^k$ into at most $\deg(f)$-many disjoint subsets; and count the number of roots in each such subset. The construction of $\mathcal{L}$ is incremental, by doing arithmetic modulo special ideals, 

{\bf Split ideals.}
A {\em split ideal} $I_l$ of length $l+1$, and degree $b$, is a `triangular' ideal defined as $I_l=\langle h_0(x_0),h_1(\bar{x}_1),\ldots,h_{l}(\bar{x}_l)\rangle$, where the notation $\bar{x}_i$ refers to the variable set $\{x_0,\ldots,x_i\}$ and $b=\prod_{0\le i\le l} \deg_{x_i}(h_i)$. It implicitly stores a size-$b$ subset of the root-set of $f\bmod p^k$, where a root looks like $\sum_{0\le i\le l} x_ip^i$ till precision $p^{l+1}$. Note that a root $r$ of $f\bmod p^k$ is also a root of $f\bmod p^l$ for all $l\in[k]$. Since we cannot access them directly, we `virtualize' them in the notation $\bar{x}_l$.

The structure of these ideals is quite nice and recursive (Section \ref{sec-pre}). So it may keep splitting (in Algorithm \ref{algo1}) till it becomes a {\em maximal ideal}, which corresponds to a single point in $(\F_p)^l$ and has degree one. Or, the algorithm may halt earlier, due to `stable clustering' of roots, and then we call the ideals-- {\em maximal split ideal}; in fact, $\mathcal{L}$ has only maximal split ideals. These do not give us the actual roots but do give us their count!

\textbf{List data structure.}
$\mathcal{L}$ implicitly stores, and may partition, the root-set of $f\bmod p^k$. Essentially, $\mathcal{L}$ is a set of at most $d$ maximal split ideals, i.e.~$\mathcal{L}=\{I_1(l_1,d_1),\ldots,I_n(l_n,d_n)$ $\}$, where each ideal $I_j \subseteq \F_p[\bar{x}_{k-1}]$ has two parameters--- length $l_j$ and degree $d_j$. A maximal split ideal $I(l,D)$ implicitly stores a size-$D$ subset of the root-set of $f\bmod p^k$. This yields a simple count of $ D p^{k-l}$ for the corresponding roots. Ideals in $\mathcal{L}$ have the property that they represent disjoint subsets of roots; and they collectively represent the whole root-set of $f\bmod p^k$. Thus, $\mathcal{L}$ gives us both the (implicit) structure and the (exact) size of the root-set of $f\bmod p^k$. In the intermediate steps of the algorithm, for efficiency reasons, we will store a tuple $(I_j, f_{I_j})$ in a changing {\em stack} $S$. Where, $f_{I_j}(\bar{x}_{l_j-1},x) := f(x_0+px_1+\cdots+p^{l_j-1}x_{l_j-1}+p^{l_j}x)\bmod \hat{I}_j$ is a `shifted and reduced' version of $f$ tagging along (with $x$ as the only {\em free} variable).

%
%

%
%
%

\textbf{Roots-tree data structure.}
Most importantly, we need to prove that $|\mathcal{L}|$, and the degree of the split ideals in $\mathcal{L}$, remains at most $\deg(f)$ at all times in the algorithm (while $f\bmod p^k$ may have exponentially many roots). To achieve this, we use a different way to look at the data structure $\mathcal{L}$--- in {\em tree} form $RT$ where each generator $h_i$ appearing in an $I\in \mathcal{L}$ appears as an edge of the tree; conversely, each tree node $v$ denotes the intermediate split ideal corresponding to the path from the root (of the tree $RT$) to $v$. 

The roots-tree $RT$ has a useful parameter at every node-- degree. Degree of a node measures the possible extensions to the next level, and it possesses the key  property: it `distributes' to its children degrees. This helps us to simultaneously bound the width of $RT$ and degree of split ideals, to be at most the degree $\deg(f)$ of the root node. Otherwise, since we compute with $k$-variate polynomials, a naive analysis of the tree-size (resp.~degree of split ideals) would give a bound of $\deg(f)^k$, or a slightly better $\deg(f)2^k$ as in \cite[pg.9]{cheng2017counting}; which is exponential in the input size $\deg(f)\cdot k\log p$.


\vspace{-1mm}
\subsection{Proof overview}\label{sec-overview}
\vspace{-1mm}


\textbf{Proof idea of Theorem \ref{thm1}.} 
Let $R:=\Z/\langle p^k\rangle$; so $R/\langle p\rangle\cong \F_p$. Let $\calZ_{R}(f)$ be the zeroset of $f\bmod p^k$.


The idea to count roots of $f\bmod p^k$ comes from the elementary fact:
Any root $r\in R$ of $f\bmod p^k$ can be seen in a $p$-adic (or base-$p$) representation as 
$r=: r_0 + p r_1+ p^2 r_2+\ldots + p^{k-1} r_{k-1}$, for each $r_i \in \F_p$.
Thus, we decompose our formal variable $x$ into multi-variables $x_0,\ldots,x_{k-1}$ being related as, $x = x_0 + p x_1+ p^2 x_2+\ldots + p^{k-1} x_{k-1}$. 

Though, getting roots of $f(x_0)\bmod p$ deterministically is difficult, we can get the count on the number of roots of $f(x_0)\bmod p$ from the degree of a polynomial $h(x_0)\in \F_p[x_0]$, which is the gcd of $f$ and {\em Frobenius} polynomial $x_0^{p}-x_0 \bmod p$. This way of implicitly representing a set of desired objects by a polynomial and using its properties (eg.~degree) to get a count on the objects is widely termed as {\em polynomial method}.

This gives us a length-$1$ and degree-$\deg_{x_0}(h_0)$ split ideal $I_0:= \langle h_0(x_0)\rangle$. Since $I_0$ represents all roots of $f\bmod p$, we can again apply the polynomial method to incrementally build on ideal $I_0$ to get greater length split ideals representing roots of $f$ with greater precision, say $\bmod\ p^{l+1}$.

To do this, we trivially lift $I_0$ to make it an ideal $\hat{I}_0$ in $R$. Solve $f(x_0+p x)\equiv p^{\alpha} g(x_0,x) \bmod \hat{I}_0$ for $\alpha\in\N$ and $g\not \equiv 0\bmod p$. Reduce $g(x_0,x)$ over $\F_p$ again, and calculate the next set of candidates for $x_1$ implicitly in a polynomial $h_1\in \F_p[x_0,x]$ defined as, $h_1:= \Call{GCD}{g(x_0,x)\bmod p, x^{p}-x} \bmod I_0$. Using the properties of split ideal (Lemma \ref{lemma-actual-gcd}), multivariate-gcd modulo $I_0$ yields $h_1$ that `stores' all the candidates for $x_1$, for each root $x_0$ represented by $I_0$. So, we get a length $2$ split ideal $I_1:=I_0+\langle h_1(x_0,x_1)\rangle$.

In every iteration, we add a new variable, by solving equations like $f(x_0+p x_1+p^2 x_2+\ldots+p^l x_l+p^{l+1} x)\equiv p^{\alpha} g(\bar{x}_l,x)$ modulo a length $l+1$ triangular ideal $\hat{I}_l$, for $\alpha\in\N$ and $g\not \equiv 0\bmod p$. This gives us the next candidate $h_{l+1}(\bar{x}_l,x):= \Call{GCD}{g(\bar{x}_l,x)\bmod p, x^{p}-x} \bmod I_{l}$; moving to a more precise split ideal. Sometimes we get that $g$ and $x^p-x$ are coprime mod $I_{l}$, those cases indicate {\em dead-end} and we stop processing those branches.
Finally we reach $\alpha=k$, which indicates {\em full precision}; and we get a maximal split ideal $I_l$ which we add to the list $\mathcal{L}$.

{\bf Division by `zero'.}
Some computations modulo a split ideal may not be possible. These cases arise only due to {\em zerodivisors}. In those cases, we will exploit the zerodivisor to split/factor the current split ideal into more split ideals of smaller degree. We can keep track of all these split ideals using a stack and keep performing the same computations iteratively. Since a split ideal has finite length, the process must terminate. The real challenge lies in proving a good bound.

\smallskip {\bf Efficiency via roots-tree.}
Now, we need to show that the algorithm to construct $\mathcal{L}$ is efficient and that $|\mathcal{L}|\le \deg(f)$ (in fact, sum of degrees of all maximal split ideals in $\mathcal{L}$ is at most $\deg(f)$). In a particular iteration, the algorithm just performs routine computations like--  reduction modulo the current split ideal $I$, inversion, zerodivisor testing, gcd, exponentiation, and computing $p$-valuations or multiplicities; which are clearly bounded by $\poly(\deg(f), k\log p, \deg(I))$ (Sections \ref{appen-reduce-valuation} \& \ref{appen-GCD-zerodiv}). It is harder to bound the number of iterations and $\deg(I)$.

To understand the number of iterations, we review the construction of $\mathcal{L}$ as the formation of a tree, which we call roots-tree $RT$. A node of $RT$ corresponds to an intermediate split ideal $I$, where an edge at level $i$ on the path from the root (of $RT$) to the node corresponds to the generator $h_i(\bar{x}_i)$ of $I$. Each time we update a split ideal $I_{l-1}$ to $I_l:=I_{l-1}+\langle h_l\rangle$ we add a child, to the node corresponding to $I_{l-1}$, hanging by a new edge labelled $h_l$. 
	Similarly, splitting of an ideal at some generator $h_i(\bar{x}_i)$ into $m$ ideals corresponds to creating $m$ subtrees hanging by edges which are $m$ copies of the edge labelled $h_i$. This way the roots-tree upper bounds the number of iterations; moreover, the maximal split ideals in $\mathcal{L}$ appear as leaves in $RT$.

{\em Degree distribution in $RT$.}
Each node $N$ of $RT$ has an associated parameter, `degree of node' $[N]$ (Definition \ref{def-deg-node}), which is defined in such a way that it distributes to degree of its children (i.e.~$[N]$ is at least the sum of degrees of its child nodes). This is intended to measure the possible extensions $x_l$ modulo the corresponding split ideal $I_{l-1}$, and is a suitable multiple of $\deg(I_{l-1})$. Applying degree's property inductively, we get that the degree of root node of $RT$, which is $\deg(f)$, distributes to the degree of the leaves and so the sum of degrees of all maximal split ideals in $\mathcal{L}$ is at most $\deg(f)$. The distributive property of $[N]$, corresponding to ideal $I_{l-1}$, comes from the fact: the degree of a child $C$ corresponding to ideal $I_l=I_{l-1}+\langle h_l\rangle$ is bounded by the {\em multiplicity} of roots of $h_{l}(\bar{a},x)$ times $\deg(I_{l-1})$, corresponding to some root $\bar{a}$ of $I_{l-1}$; and the overall sum of these multiplicities for every child of $N$ is naturally bounded by the degree  of $N$ (Lemma \ref{lemma-property-of-degree}).

The details are given in Section \ref{sec-main1}.

\medskip\noindent
\textbf{Proof idea of Theorem \ref{thm2}.} 
The idea, and even the algebra, is the same as for Theorem \ref{thm1}. The definition of list $\mathcal{L}$ easily extends to implicitly store all the basic-irreducible factors of $f\bmod p^k$ of some degree $b$ (a generalization over roots which corresponds to degree $b=1$ basic-irreducible factors). This uses a strong property possessed by basic-irreducible factors. A basic-irreducible factor $g(x)\in (\Z/\langle p^k\rangle)[x]$ of $f \bmod p^k$, of degree $b$, completely splits over the {\em Galois ring} $G(p^k,b):= \Z[y]/\langle p^k, \varphi(y)\rangle$, where $\varphi(y) \bmod p$ is an irreducible of degree $b$ (Section \ref{appen-prelim}). Conversely, if we find a root of $f(x)$, in $G(p^k,b)$, then we find a degree-$b$ basic-irreducible factor of $f\bmod p^k$.

By distinct degree factorization we can assume $f(x)\equiv (\varphi_1 \ldots \varphi_m)^e + p h(x) \bmod p^k$, where each $\varphi_i(x) \bmod p$ is irreducible and degree-$b$. We construct $\mathcal{L}$ by applying the algorithm of Theorem \ref{thm1}, with one change: every time to update a length-$l$ split ideal $I_{l-1}$ to a length $l+1$ ideal $I_l:=I_{l-1}+\langle h_l\rangle$, we compute $h_l$ using the Frobenius polynomial $x^{q}-x \bmod p$, where $q:=p^b$. Basically, for $x$, we focus on $\F_q$-roots instead of the erstwhile $\F_p$-roots.

We count the number of (distinct, monic, degree-$b$) basic-irreducible factors represented by each maximal split ideal $I(l,D)\in \mathcal{L}$ as: $Dq^{k-l}/b$. The details are given in Section \ref{sec-main2}.
 
%
%

\section{Preliminaries}
\label{sec-pre}
\vspace{-1mm}

Here we introduce our main tool - `split ideals'. Proofs for this section have been moved to Section \ref{app-split-proofs}. Basic introduction to Galois rings (i.e.~non-prime characteristic analog of finite fields), Hensel lifting, randomized factoring over finite fields, etc.~have been moved to Section \ref{appen-prelim}.

We will be given a univariate polynomial $f(x)\in \mathbb{Z}[x]$ of degree $d$ and a prime power $p^k$ (for a prime $p$ and a positive integer $k\in \mathbb{N}$). 
Wlog, we assume that $f$ is monic over $\mathbb{F}_p$.

A tuple of variables $(x_0,\ldots,x_l)$ will be denoted by $\bar{x}_l$. Often, an $(l+1)$-variate polynomial $a(x_0,x_1,\ldots,x_l)$ will be written as $a(\bar{x}_l)$, 
and the polynomial ring $\F_p[x_0,\ldots,x_l]$ as $\F_p[\bar{x}_l]$. 

We denote the ring $\mathbb{Z}/\langle p^k\rangle$ by $R$ (ring $R/\langle p\rangle$ is the same as field $\mathbb{F}_p$). 
An element $a\in R$ can be seen in its $p$-adic representation as $a= a_0+ p a_1+\ldots+p^{k-1}a_{k-1}$, where $a_i \in \F_p$ for $i\in \{0,\ldots,k-1\}$.

$\mathcal{Z}_R(g):= \{ r\in R \mid g(r)\equiv 0 \bmod p^k \}$ denotes the {\em zeroset} of a polynomial $g(x) \in R[x]$.

{\em Zeroset of an ideal} $I \subseteq \F_p[x_0,\ldots,x_l]$ is defined as the intersection of zeroset of all polynomials in $I$,
$\mathcal{Z}_{\mathbb{F}_p}(I):= \{ \bar{a}=(a_0,\ldots,a_l) \in (\mathbb{F}_p)^{l+1} \mid g(\bar{a}) \equiv 0 \bmod p, \forall g \in I \}$.

We will heavily use ideals of the form $I:= \langle h_0(\bar{x}_0),h_1(\bar{x}_1),\ldots,h_l(\bar{x}_l) \rangle$ satisfying the condition--- for any $i\in[l+1]$ and 
$\bar{a} \in \mathcal{Z}_{\mathbb{F}_p}(\langle h_0(\bar{x}_0),h_1(\bar{x}_1),\ldots,h_{i-1}(\bar{x}_{i-1}) \rangle)$, polynomial $h_{i}(\bar{a}, x_i)$ splits completely into distinct linear factors. 
They are formally defined as:

\begin{definition}[Split ideal]
\label{def-split-ideal}
We will call a polynomial {\em monic wrt $x$} if the leading-coefficient is one.
Given $f(x)\in R[x]$, an ideal $I$, in $\mathbb{F}_p[\bar{x}_l]$, is called a {\em split ideal wrt $f\bmod p^k$} if,
\\\noindent
1) $I$ is a {\em triangular ideal of length $l+1$}, meaning: $I=: \langle h_0(\bar{x}_0),h_1(\bar{x}_1),\ldots,h_l(\bar{x}_l) \rangle$, for some $0\le l\leq k-1$; $h_i(\bar{x}_i)\in \F_p[\bar{x}_i]$ is monic wrt $x_i$, for all $ i \in \{0,\ldots,l\}$,
\\\noindent
2) $|\mathcal{Z}_{\mathbb{F}_p}(I)| = \prod_{i= 0}^{l} \deg_{x_i}(h_i)$, and
\\\noindent
3) $\forall (a_0,\ldots,a_l) \in \mathcal{Z}_{\F_p}(I)$, $f(a_0+p a_1+\ldots+p^l a_l)\equiv 0 \bmod p^{l+1}$.

\smallskip
The {\em length of $I$} is $l+1$ and its {\em degree} is $\deg(I) := \prod_{i= 0}^{l} \deg_{x_i}(h_i)$.
\end{definition}

Split ideal $I$ relates to possible roots of $f\bmod p^k$. Since $f,p,k$ are fixed, we will call $I$ a {\em split ideal}. 
The definition of a split ideal implies that its roots represent a set of ``potential'' roots of $f$, i.e.~roots of $f$ modulo some $p^{l+1}$ for $0\le l<k$.
Restriction of a split ideal is also a split ideal. 

\begin{lemma}[Restriction of a split ideal]
\label{lemma-split-ideal-is-recursive}
Let $I_l:=\langle h_0(\bar{x}_0),\ldots,h_l(\bar{x}_l) \rangle$ be a split ideal in $\F_p[x_0,\ldots,x_l]$, then ideal $I_{j}:=\langle h_0(\bar{x}_0),\ldots,h_{j}(\bar{x}_{j}) \rangle$ is 
also a split ideal in $\F_p[x_0,\ldots,x_{j}]$, for all $0\le j\leq l$.
\end{lemma}

Further, we show that a split ideal $I$ can be decomposed in terms of its zeros.

\begin{lemma}[Split ideal structure]
\label{lemma-split-ideal-str}
A split ideal $I \subseteq \F_p[x_0,\ldots,x_l]$ can be decomposed as $I=\bigcap_{\bar{a} \in \calZ_{\F_p}(I)} I_{\bar{a}}$, where each $I_{\bar{a}}:=\langle x_0-a_0,\ldots,x_l-a_l \rangle$ corresponds to root $\bar{a} =: (a_0,\ldots,a_l) \in \calZ_{\F_p}(I)$. By Chinese remainder theorem, $  R/I = \bigoplus_{\bar{a} \in \calZ_{\F_p}(I)} R/I_{\bar{a}}$ .
\end{lemma}

Let $I =: \langle h_0(\bar{x}_0),h_1(\bar{x}_1),\ldots,h_{l}(\bar{x}_l) \rangle$ be a split ideal. Suppose some $h_i$ factors as $h_i(\bar{x}_i)= h_{i,1}(\bar{x}_i)\ldots h_{i,m}(\bar{x}_i)$.
Define $I_j:= \langle h_0(\bar{x}_0), \ldots, h_{i-1}(\bar{x}_{i-1}), h_{i,j}(\bar{x}_{i}), h_{i+1}(\bar{x}_{i+1}), \ldots, h_{l}(\bar{x}_{l}) \rangle$, for $ j\in [m]$.
The following corollary of Lemma~\ref{lemma-split-ideal-str} is evident because root-sets of $I_j$ partition the root-set of $I$. 

\begin{corollary}[Splitting split ideals]
\label{cor-split-ideal-str}
Let $I=\langle h_0(\bar{x}_0), \ldots, h_{l}(\bar{x}_l) \rangle$ be a split ideal of $\F_p[x_0,\ldots,x_l]$. 
Let some $h_i(\bar{x}_i)$ factor as $h_i(\bar{x}_i)= h_{i,1}(\bar{x}_i)\ldots h_{i,m}(\bar{x}_i)$. 

Then, $I=\bigcap_{j=1}^{m} I_j$, 
where each $I_j :=\langle h_0(\bar{x}_0), \ldots, h_{i-1}(\bar{x}_{i-1}), h_{i,j}(\bar{x}_{i}), h_{i+1}(\bar{x}_{i+1}), \ldots, h_{l}(\bar{x}_{l}) \rangle$ is a split ideal. 
\end{corollary}

We call a split ideal $I_l:= \langle h_0,\ldots,h_{l}\rangle$ to be \emph{maximal split ideal} if,  
\\\noindent
1) for any $\bar{a}=(a_0,\ldots,a_l) \in \calZ_{\F_p}(I_l)$, $g(x):= f(a_0+pa_1+\ldots+p^la_{l}+p^{l+1} x)$ vanishes identically mod $p^k$,
\\\noindent
2) the restriction $I_{l-1}:=\langle h_0,\ldots,h_{l-1}\rangle$ does not follow the previous condition.



\begin{lemma}[Roots represented by a root of maximal split ideal]
\label{lemma-MSI-root-represent}
Let $I$ be a maximal split ideal of length $l+1$, then a zero $\bar{a}=(a_0,\ldots,a_l) \in \calZ_{\F_p}(I)$ maps to exactly $p^{k-l-1}$ zeros of $f$ in $\calZ_R(f)$. 
We will say that these $p^{k-l-1}$ roots of $f$ are \emph{represented} by $\bar{a}$. 
\end{lemma}

%
%



\section{Proof of Theorem~\ref{thm1}} \label{sec-main1}
\vspace{-1mm}

The algorithm to compute a compact data-structure which stores roots of $f \bmod p^k$ will be described in Section~\ref{sec-algo-thm1}.
Algorithm's correctness will be proved in Section~\ref{sec-pf-algo1}, which involves studying the algebraic structure underlying the algorithm. Its efficiency will be shown in Section~\ref{sec-pf-thm1}, by devising an auxiliary structure called roots-tree and the important notion of `degree of a node'.

%

\subsection{Algorithm to implicitly partition the root-set of $f(x)\bmod p^k$}
\label{sec-algo-thm1}
\vspace{-1mm}


We describe our algorithm in this section. It takes a monic univariate polynomial $f(x) \in \mathbb{Z}[x]$ of degree $d$ and a prime-power $p^k$ as input (in binary), 
and outputs a list of at most $d$ maximal split ideals whose roots partition the root-set of $f$ modulo $p^k$. 

A maximal split ideal $I_j=:\langle h_0(\bar{x}_0),\ldots,h_l(\bar{x}_l)\rangle$ has $|\calZ_{\F_p}(I_j)|=\prod_{i=0}^{l} \deg_{x_i}(h_i)$ zeros, 
and each such zero `represents' $p^{k-l-1}$ actual zeros of $f\bmod p^k$ (Lemma \ref{lemma-MSI-root-represent}). Thus, this algorithm gives an exact count on the number of zeros of $f$ in $R$.




%
%
%

\textbf{Overview of Algorithm~\ref{algo1}: }
Since any root of $f$ mod $p^k$ is an extension of a root modulo $p$, the algorithm starts by initializing a stack $S$ with the ideal $I:= \langle h_0(x_0)\rangle$, where $h_0(x_0) := \gcd(x_0^p-x_0, f(x_0))$. This is a split ideal containing all the roots of $f$ mod $p$. By a {\em lift} $\hat{I}\subset R[x_0]$ of $I$, we mean the ideal generated by the generator $\{h_0\}$ when viewed as a polynomial in $R[x_0]$ (i.e.~char $p^k$).

At every intermediate iteration (Steps $4-21$), we {\em pop} a split ideal from the stack and \emph{try} to increase the precision of its root-set (equivalently, lengthen the split ideal).
This step mostly results in two cases: either we succeed and get a split ideal whose root-set has increased precision (Step $18$) by a new placeholder $x_{l+1}$, or the split ideal factors into more split ideals increasing the size of the stack $S$ (Steps $10,14,20$). We update the relevant `part of $f$' to $f_I(\bar{x}_l, x_{l+1}+p x) \bmod \hat{J}$ ($J$ is the new split ideal) that we carry around with each split ideal. This helps in efficiently increasing the precision of roots in the next iteration. Otherwise, computing $f\left(x_0+px_1+\cdots+ p^l x_l+ p^{l+1} x\right)/p^\alpha \bmod I$ is too expensive, in Step $6$, due to the underlying degree-$d$ $(l+1)$-variate monomials blowup.

If we reach a maximal split ideal (Step $7$), it is moved to a list $\mathcal{L}$. Sometimes the split ideal cannot be extended and we get a {\em dead-end} (Step $16$). 
The size of the stack decreases when we get a maximal split ideal or a dead-end. The algorithm terminates when stack becomes empty.  
List $\mathcal{L}$ contains maximal split ideals which partition, and cover, the root-set of $f$ (implicitly). This becomes our output.

The main intuition behind our algorithm: If two roots of a split ideal (representing potential roots of $f$) give rise to different number of roots of $f$, the split ideal will get factored further.
Though not at all apparent immediately, we will show that the algorithm takes only polynomial number of steps (Section~\ref{sec-pf-thm1}).
	
%
%
%
%

We will use four subroutines to perform standard ring arithmetic modulo split ideals; they are described in the Appendices \ref{appen-reduce-valuation} \& \ref{appen-GCD-zerodiv}. 
\begin{enumerate}
\item Modify $f$ (Steps 3, 18, 20) whenever pushing  in the stack (Lemma~\ref{lemma-stack-update} \&  \ref{lemma-split-reduce}).
\item \Call{Reduce}{$a(\bar{x}_l), J_{l}$} gives the reduced form of $a$ mod triangular ideal $J_l$ (over a Galois ring).
\item \Call{Test-Zero-Div}{$a(\bar{x}_l), I_l$} either reports that $a$ is a not a zero-divisor modulo triangular ideal $I_l$ or outputs a non-trivial factorization of one of the generators of $I_l$ when true.
\item \Call{GCD}{$a(\bar{x}_l,x), b(\bar{x}_{l},x), I_{l}$} either successfully computes a monic gcd, wrt $x$, of two multivariates modulo a triangular ideal $I_{l}$, or encounters a zerodivisor in intermediate computation (outputting $False$ and a non-trivial factorization of one of the generators of $I_{l}$). 
\end{enumerate}


\begin{breakablealgorithm}
 \caption{Root-counting mod $p^k$}
 \label{algo1}
\begin{algorithmic}[1]
\parState{Let $\mathcal{L}=\{ \}$ be a list and $S=\{\}$ be a stack (both initially empty).}
\parState{Let $\tilde{f}(x_0) :=f(x_0) \bmod p$ for a monic univariate $\tilde{f}\in \mathbb{F}_p[x_0]$ of degree $d$.}
\parState{[\textbf{Initializing the stack S}] Let $h_0(x_0) := \gcd(\tilde{f}(x_0),x_0^p-x_0)$, $I:=\langle h_0 \rangle$, $\hat{I} \subseteq R[x_0]$ be a lift of $I$. Compute $f_I(x_0,x):= f(x_0+ p x) \bmod \hat{I}$ using Lemma \ref{lemma-stack-update}. Update $S \leftarrow push( (\{h_0\}, f_I) )$.}
\While{$S$ is not empty}
\parState{$S_{top}\leftarrow pop(S)$. Let $S_{top}=(\{h_0(x_0),\ldots,h_l(x_0,\ldots,x_l) \}, f_{I}(\bar{x}_l,x))$ where $I=\langle h_0,\ldots,h_l \rangle \subseteq \F_p[\bar{x}_l]$ is a split ideal.
Let $\hat{I} \subseteq R[x_0,\ldots,x_l]$ be a lift of $I$.}
\parState{[\textbf{Valuation computation}] Compute $\alpha\in\N$ and $g \in R[\bar{x}_l,x]$ such that $f_{I} \equiv p^{\alpha} g(\bar{x}_l,x) \bmod \hat{I}$ and $p\not| g \bmod \hat{I}$.}

\parState{[\textbf{Maximal split ideal found}] \textbf{if}($\alpha \geq k$) \textbf{then} update List $\mathcal{L}\leftarrow \mathcal{L} \cup \{ I\}$. Go to Step $4$.}
\parState{Let $\tilde{g}:= g(\bar{x}_l,x)\bmod I$ be the polynomial in $\mathbb{F}_p[\bar{x}_{l},x]$, and let $g_1(\bar{x}_l)$ be the leading coefficient of $\tilde{g}(\bar{x}_l,x)$ wrt $x$.}
\If{\Call{Test-Zero-Div}{$g_1(\bar{x}_l)$, $I$}$=True$}
\parState{$\Call{Test-Zero-Div}{g_1(\bar{x}_l), I}$ returns a factorization $h_i(\bar{x}_i)=: h_{i,1}(\bar{x}_i)h_{i,2}(\bar{x}_i)\ldots h_{i,m}(\bar{x}_i)$ $ \bmod I_{i-1}$ of some generator $h_i(\bar{x}_i)$ of $I$. Go to Step $20$.}
\EndIf

[\textbf{Filter out distinct virtual $\F_p$-roots by taking gcd with $x^{p}-x$}]
\parState{
Recompute $\tilde{g}:= g(\bar{x}_l,x)\cdot g_1(\bar{x}_l)^{-1}\bmod I$ (Lemmas \ref{lem-div-mod-I}, \ref{lemma-reduction}). Compute $x^{p}$ by repeatedly squaring and reducing modulo the triangular ideal $I+\langle \tilde{g}\rangle$ (Algorithm \ref{algo-reduction} and Lemma \ref{lemma-reduction}). This yields $\tilde{h}_{l+1}(\bar{x}_l,x):=x^{p}-x \bmod I$ in a reduced form.}
%
\If{\Call{GCD}{$\tilde{g}$, $\tilde{h}_{l+1}$, $I$} $=False$}
\parState{The call \Call{GCD}{$\tilde{g}$, $\tilde{h}_{l+1}$, $I$} returns factorization $h_i(\bar{x}_i)=h_{i,1}(\bar{x}_i)h_{i,2}(\bar{x}_i)\ldots h_{i,m}(\bar{x}_i) \bmod I_{i-1}$ of a generator $h_i(\bar{x}_i)$ of $I$. Go to Step  $20$.}

\ElsIf{$\tilde{g}$ and $\tilde{h}_{l+1}$ are coprime}
\parState{[\textbf{Dead End}] The ideal $I$ cannot grow more, go to Step $4$.}
\Else
\parState{[\textbf{Grow the split ideal $I$}] Here $\gcd_x(\tilde{g}, \tilde{h}_{l+1}) \bmod I$ is non-trivial, say $h_{l+1}(\bar{x}_l,x)$ (monic wrt $x$). Substitute $x$ by $x_{l+1}$ in $h_{l+1}(\bar{x}_l,x)$ and update $J\leftarrow I+\langle h_{l+1}(\bar{x}_{l+1})\rangle$.

Let $\hat{J} \subseteq R[x_0,\ldots,x_{l+1}]$ be a lift of $J$.
 Substitute $x$ by $x_{l+1}+p x$ in $f_I(\bar{x}_l,x)$, and compute $f_J(\bar{x}_{l+1},x):= f_I(\bar{x}_l, x_{l+1}+p x) \bmod \hat{J}$ using Lemma \ref{lemma-stack-update}. Update $S\leftarrow push((\{h_0,\ldots,h_{l+1}\}, f_J))$, and go to Step $4$.}
\EndIf

\parState{[\textbf{Factoring split ideals}] We have a factorization $h_i(\bar{x}_i)= h_{i,1}(\bar{x}_i)h_{i,2}(\bar{x}_i)\ldots h_{i,m}(\bar{x}_i) \bmod I_{i-1}$ of a generator $h_i$ of $I$. 
Push $S_{top}$ back in stack $S$. For every entry $(U,f_{\langle U\rangle}) \in S$, where $h_i(\bar{x}_i)$ appears in $U$, find $m$ (smaller) split ideals $U_j$ (using Corollary \ref{cor-split-ideal-str}); using Lemma \ref{lemma-split-reduce} compute $f_{\langle U_j\rangle}$ and push $(U_j, f_{\langle U_j\rangle})$ in $S$, for $j\in[m]$. 
}


\EndWhile
\parState{Return $\mathcal{L}$ (the list of maximal split ideals partitioning the root-set $\calZ_R(f)$).}
\end{algorithmic}
\end{breakablealgorithm}

\subsection{Correctness of Algorithm~\ref{algo1}}
\label{sec-pf-algo1}
\vspace{-1mm}


Our main goal is to prove the following result about partitioning of root-set.

\begin{theorem}[Algo \ref{algo1} partitions $\calZ_{R}(f)$]
\label{thm-proof-algo1}
Algorithm \ref{algo1} yields the structure of the root-set $\calZ_{R}(f)$ through a list data structure $\mathcal{L}$ (a collection of maximal split ideals $I_1,\ldots,I_n$) which partitions the zeroset $\calZ_{R}(f) =: \bigsqcup_{j \in [n]}S_j $,
where $S_j$ is the set of roots of $f\bmod p^k$ represented by $\calZ_{\F_p}(I_j)$.
\end{theorem}

Later, we will show a surprising property: $n \leq d$ (Section \ref{sec-pf-thm1}).

%

\begin{proof}[Proof of Theorem \ref{thm-proof-algo1}]
From Lemmas~\ref{lemma-content-stack}, \ref{lemma-algo1-terminates} and the definition of maximal split ideal, it is clear that Algorithm \ref{algo1} returns a list $\mathcal{L}$ containing maximal split ideals $I_1,\ldots,I_n$, for $n\in \N$.
Further, we show:
\\\noindent
1) The root-set of $I_j$ ($1\leq j \leq n$) yields a subset $S_j$ of $\calZ_{R}(f)$, and they are pairwise disjoint.
\\\noindent
2) Given a root $r\in \calZ_{R}(f)$, there exists $j$ such that $r$ is represented by a root in $\calZ_{\F_p}(I_j)$. 

\smallskip
For the first part, root-sets for different maximal split ideals $I_j$ are pairwise disjoint because of Lemma~\ref{lemma-content-stack}. Each of these root-set yields a subset of the zeroset of $f \bmod p^k$ (follows from the definition of maximal split ideal).

For the second part, let $r=: \sum_{i=0}^{k-1} r_i p^i$ be a root in $\calZ_{R}(f)$. Stack $S$ was initialized by the split ideal $\langle h_0 := \gcd(f(x_0)\bmod p, x_{0}^{p}-x_0)\rangle$; so $r_0 \in \calZ_{\F_p}(I_0)$, as $f(r_0) \equiv f(r)\equiv0 \bmod p$. 

Assume that $I_0$ is not a maximal split ideal (otherwise we are done).
Applying Lemma~\ref{lemma-root-extension}, there must exist an $I_1$ whose root-set contains $(r_0,r_1)$. 
Repeated applications of Lemma~\ref{lemma-root-extension} show that we will keep getting split ideals of larger lengths, partially representing $r$; finally, reaching a maximal split ideal (say $I_j$) fully representing $r$. 


We showed that each root $r$ of $f\bmod p^k$ is represented by a {\em unique} maximal split ideal $I$, given by Algorithm~\ref{algo1}, and they collectively represent exactly the roots of $f$ modulo $p^k$. 
Hence, root-sets of ideals in $\mathcal{L}$ partition the zeroset $\calZ_{R}(f)$. 
\end{proof}

Now, let us see the properties of our algorithm which go in proving Theorem~\ref{thm-proof-algo1}.
Given a polynomial $g(\bar{x}_l) \in \F_p[\bar{x}_l]$ and an element $\bar{a} \in \F_p^l$, consider the {\em projection} $g_{\bar{a}}(x_l) := g(\bar{a},x_l)$. 
Using Chinese remainder theorem (Lemma~\ref{lemma-split-ideal-str}) we easily get the following degree condition. (Here, $\text{lc}_{x}$ refers to the leading coefficient wrt variable $x$.)

\begin{claim}
\label{claim-monic-degree}
Let $I$ be a split ideal of $\F_p[\bar{x}_{l-1}]$ and $g \in \F_p[\bar{x}_l]$. Then, $\text{lc}_{x_l}(g)$ is unit mod $I$ ~iff
$ \forall \bar{a} \in \calZ_{\F_p}(I), \ \deg(g_{\bar{a}}(x_l)) \ =\  \deg_{x_l}(g(\bar{x}_l) \bmod I) .$
\end{claim}

Chinese remaindering also gives us a gcd property under projections.

\begin{lemma}
\label{lemma-actual-gcd}
Let $w(\bar{x}_l),z(\bar{x}_l)\in \F_p[\bar{x}_l]$ and $I_{l-1} \subseteq \F_p[\bar{x}_{l-1}]$ be a split ideal. Suppose Algorithm \ref{algo-gcd} succeeds in computing gcd of $w$ and $z$ mod $I_{l-1}$:
define $h(\bar{x}_l):=\Call{GCD}{w(\bar{x}_l),z(\bar{x}_l),I_{l-1}}$. 
	Then, for all $\bar{a} \in \mathcal{Z}_{\F_p}(I_{l-1})$: 
	$h_{\bar{a}}(x_l)$ equals $\gcd(w_{\bar{a}}(x_l), z_{\bar{a}}(x_l))$ up to a unit multiple (in $\F^*_p$).
\end{lemma}
\begin{proof}
Lemma~\ref{lemma-gcd} proves, $h(\bar{x}_l)$ is a monic polynomial mod $I_{l-1}$, s.t., $h | w $ and $h | z $ (mod $I_{l-1}$). Fix $\bar{a} \in \mathcal{Z}_{\F_p}(I_{l-1})$.
Since $h_{\bar{a}}(x_l) \not\equiv 0 \bmod p$ ($\because h$ is monic), restricting $\bar{x}_{l-1}$ to $\bar{a}$ gives $h_{\bar{a}} | w_{\bar{a}}$ and $h_{\bar{a}} | z_{\bar{a}}$, showing $h_{\bar{a}} | \gcd(w_{\bar{a}}, z_{\bar{a}})$, in $\F_p[x_l]$. 

Lemma \ref{lemma-gcd} also shows that there exists $u, v \in (\F_p[\bar{x}_{l-1}]/I_{l-1})[x_l]$, such that, $h = u w + v z$. 
Restricting first $l$ co-ordinates to $\bar{a}$, we get $h_{\bar{a}} = u_{\bar{a}} w_{\bar{a}} + v_{\bar{a}} z_{\bar{a}}$. This equation implies $\gcd(w_{\bar{a}}, z_{\bar{a}})| h_{\bar{a}}$. Thus, we get an equality up to a unit multiple.
\end{proof}

Let $I\subseteq \F_p[\bar{x}_i], J \subseteq \F_p[\bar{x}_j]$ be two split ideals (say $i \leq j$). 
$I$ and $J$ are called \emph{prefix-free} iff
$ \nexists \ \bar{a} = (a_0, a_1, \ldots, a_i) \in \calZ_{\F_p}(I), \ \bar{b} = (b_0, b_1,\ldots, b_j) \in \calZ_{\F_p}(J): a_k = b_k \ \forall k \leq i$. 

(Note that it may still happen that $(a_0,\ldots, a_{i-1}) = (b_0,\ldots, b_{i-1})$ above.)

Our next lemma shows an invariant about Algorithm~\ref{algo1}.

\begin{lemma}[Stack contents]
\label{lemma-content-stack}
Stack $S$ in Algorithm~\ref{algo1} satisfies following conditions at every point:
\\\noindent 1) $l<k$ and in Step 6, $\alpha>l$. 
\\\noindent 2) All ideals in $S$ are split ideals.
\\\noindent 3) Any two ideals in $S$ are prefix-free.
\end{lemma}
\begin{proof}
We first prove the invariant $1$. Step $6$ defines $g$ via $f_I$ as, $f_{I} =: p^{\alpha} g(\bar{x}_l,x) \bmod \hat{I}$. Looking at the $f_{I}$ analogues pushed in Steps $3, 18, 20$, one easily deduces the invariants: 

\noindent $f\left(\sum_{0\le i\le l} x_i p^i +xp^{l+1} \right) \,\equiv\, f_{I}(\bar{x}_l,x)  \bmod \hat{I}$, and 

\noindent $f\left(\sum_{0\le i\le l} x_i p^i\right) \,\equiv\, 0 \bmod \hat{I}+\langle p^{l+1}\rangle$ . 

Thus, $f\left(\sum_{0\le i\le l} x_i p^i\right) \equiv p^{\alpha} g(\bar{x}_l,x) \equiv 0  \bmod \hat{I}+\langle p^{l+1}\rangle$. Since, $p\nmid g \bmod \hat{I}$, we deduce $\alpha>l$. Moreover, by Step 7 we know that $l<k$ throughout the algorithm. 

There are three ways in which a new ideal is added to stack $S$. We show below that the invariant is maintained in all three cases.

\smallskip (Step $3$) $S$ is initialized with the ideal $I=\langle h_0(x_0) \rangle \subseteq \F_p[x_0]$. The triangular ideal $I$ is a split ideal, 
because $|\calZ_{\F_p}(I)|=\deg_{x_0}(h_0)$ and its root are all the distinct roots of $f(x_0)\bmod p$.

\smallskip (Step $20$) Ideal $I_l$ is popped from $S$, and some generator $h_i$ of $I_l$ splits. In this case, we update $S$ with the corresponding factors of any $(U, f_{\langle U\rangle})\in S$, wherever currently $U$ has $h_i$. 
Corollary~\ref{cor-split-ideal-str} shows that the factors of $U$ are split ideals themselves, and their root-sets partition that of $U$. 
Thus, these root-sets are prefix-free among themselves. Moreover, they are prefix-free with any other ideal $J$ appearing in $S$, because $U$ was prefix-free with $J$.   

\smallskip (Step $18$) Ideal $I_l$ is popped, it grows to $I_{l+1}$ by including $h_{l+1}(\bar{x}_{l}, x) = \gcd_x(\tilde{g}(\bar{x}_{l}, x), x^p-x) \bmod I_l$ ($\tilde{g}$ is defined in Step 8). 
First (resp.~third) condition for $I_{l+1}$ being a split ideal follows from the definition of $\tilde{g}$ (resp.~$h_{l+1}$).

For the second condition for $I_{l+1}$ being a split ideal, fix a particular root $\bar{a} \in \calZ_{\F_p}(I_l)$. 
Using Lemma~\ref{lemma-actual-gcd}, the projection $h_{l+1,\bar{a}}(x)$ equals  $\gcd(\tilde{g}_{\bar{a}}(x), x^p-x)$ (up to a unit multiple). 
By Lemma~\ref{lemma-gcd}, $h_{l+1}$ is monic mod $I_l$; giving $\deg(h_{l+1,\bar{a}}) = \deg_{x_{l+1}}(h_{l+1})$. Since $h_{l+1}| x^p-x$, 
there are exactly $\deg_{x}(h_{l+1})$-many $a_{l+1}\in \F_p$, such that $h_{l+1,\bar{a}}(a_{l+1})\equiv 0\bmod p$.
So, every root $\bar{a} \in \calZ_{\F_p}(I_l)$ can be extended to $\deg_{x}(h_{l+1})$-many roots; giving $|\calZ_{\F_p}(I_{l+1})|= \deg_{x}(h_{l+1}) \cdot\prod_{i=0}^{l} \deg_{x_i}(h_i)$. This makes $I_{l+1}$ a split ideal.

$I_{l+1}$ remains prefix-free with any other ideal $J$ of $S$, because roots of $I_{l+1}$ are extension of roots of $I_l$ (recall: $I_l$ was prefix-free with $J$ and it was popped out of $S$). 

\smallskip
This proves all the invariants for the stack $S$.
\end{proof}

Using the invariant, we prove that Algorithm~\ref{algo1} terminates on any input.

\begin{lemma}
\label{lemma-algo1-terminates}
Algorithm~\ref{algo1} finishes in finite number of steps for any $f\in \Z[x]$ and a prime power $p^k$.
\end{lemma}
\begin{proof}
We show that the number of iterations in Algorithm~\ref{algo1} are finite. Assume that all the ideals which result in a dead-end are moved to a list $D$; say $C$ is the disjoint union of all ideals in $S$, $\mathcal{L}$ and $D$. Whenever a split ideal $I$ from $S$ is moved to $\mathcal{L}$ or $D$, the underlying roots (of $I$) stop extending to the next precision.
Togetherwith Lemma~\ref{lemma-content-stack}, we deduce that in fact all the ideals in $C$ are prefix-free. Now by Step 18, and the rate of growth of split ideals up to length $l+1\le k$, we get a lazy estimate of $|C|\le \min(d^k, p^k)$.

Let $\len(I)$ denote the length of an ideal $I$, it is bounded by $k$. Notice that factoring/growing an ideal increases $\sum_{I\in C} \len(I)$; and getting a maximal split ideal/ dead-end increases $|\mathcal{L}| + |D|$. Thus, every iteration of the algorithm strictly increases the quantity $(\sum_{I\in C} \len(I)) + |\mathcal{L}| + |D|$. 
By the estimate on $|C|$, all the terms in this quantity are bounded; thus, the number of iterations are finite.
\end{proof}

The following lemma shows: if we see a restriction of $r \in \calZ_{R}(f)$ (say, up to length $l+1$) at some point in Algorithm~\ref{algo1}, 
we will again see its restriction of length $l+2$ at a later point in the algorithm.

\begin{lemma}[Getting roots with more precision]
\label{lemma-root-extension}
Assume that at some time (say $t$), Algorithm~\ref{algo1} pops an ideal $I$ of length $l+1$, that is not yet a maximal split ideal.  
Let $\bar{a}=(a_0,\ldots,a_l) \in \calZ_{\F_p}(I)$ partially represent a ``root'' $r=: \sum_{0\le i\le l+1} a_i p^i$ such that $f(r) \equiv 0 \mod p^{l'}$, but $f(r-a_{l+1}p^{l+1}) \not\equiv 0 \mod p^{l'}$, for some $l+2\le l'\le k$.
Then, there exists a time $t'>t$, when stack $S$ will pop an ideal $J$ of length $l+2$, such that, $(\bar{a},a_{l+1}) \in \calZ_{\F_p}(J)$.
\end{lemma}

\begin{proof}
We again consider three possible situations.

\smallskip (Step $18$) Ideal $I$ grows to another split ideal, say $J$. Notice, $J$ is obtained by adding $h_{l+1}:=\Call{GCD}{g(\bar{x}_l,x),x^p-x} \bmod I$ to $I$ (setting $x\mapsto x_{l+1}$). 

Step $6$ defines $g$ via $f_I$ as, $f_{I} =: p^{\alpha} g(\bar{x}_l,x) \bmod \hat{I}$. Looking at the $f_{I}$ analogues pushed in Steps $3, 18, 20$, one can deduce the invariant: $f\left(\sum_{0\le i\le l} x_i p^i +xp^{l+1} \right) \equiv f_{I}(\bar{x}_l,x) \bmod \hat{I}$. 


Now, let us project to (suitable integral lifts of) $\bar{a}$ and consider $f\left(\sum_{0\le i\le l} a_i p^i + xp^{l+1}\right) \equiv f_{I}(\bar{a},x) \equiv p^{\alpha} g(\bar{a},x) \bmod \hat{I}$. By Step $9$, and Claim \ref{claim-monic-degree}, we are assured that $g(\bar{a},x)$, $g(\bar{x}_l,x) \bmod I$ are  equi-degree (wrt $x$). Thus, by non-maximality hypothesis we have $\alpha<l'$. Hypothesis tells us that $f\left(\sum_{0\le i\le l+1} a_i p^i\right) \equiv 0 \mod p^{l'}$. So, by the previous paragraph, $p^{\alpha} g(\bar{a}, a_{l+1}) \equiv 0 \mod p^{l'}$. Whence, 
$g(\bar{a},a_{l+1})\equiv0 \bmod p$. Clearly, $a_{l+1}^p-a_{l+1}\equiv0 \bmod p$. 
Thus, $h_{l+1}(\bar{a}, a_{l+1}) \equiv0 \bmod p$. So $(\bar{a},a_{l+1})$ is a root of $J$.

\smallskip (Step $16$) Proof of the previous case shows that $h_{l+1}(\bar{a}, x)$ has degree at least $1$, so $I$ could not result in a {\em dead-end}.

\smallskip (Step $20$) Ideal $I$ factors into (smaller) split ideals. In this case, $\bar{a}$ will be included in exactly one of those ideals (by Corollary~\ref{cor-split-ideal-str}). This ideal will be handled later in the algorithm and will give an ideal $J$ with $(\bar{a},a_{l+1})$ as root. 
\end{proof}

\subsection{Time complexity of Algorithm~\ref{algo1}--- introducing roots-tree $RT$}
\label{sec-pf-thm1}
\vspace{-1mm}


We know that Algorithm~\ref{algo1} takes finite amount of time and terminates (Lemma~\ref{lemma-algo1-terminates}). 
To show that it is efficient, note that the time complexity of the algorithm can be divided into two parts.

 1) Number of iterations taken by Algorithm~\ref{algo1}, which is clearly bounded by the number of updates on Stack $S$ in the algorithm.

 2) Time taken by the various algebraic operations in one iteration of the algorithm: reduction by a triangular ideal, valuation computation modulo a split ideal, testing if some polynomial is a zerodivisor modulo a split ideal, performing repeated squaring modulo a triangular ideal and computing gcd of two multi-variates modulo a split ideal. 

\smallskip
For the purpose of bounding iterations, we define a `virtual' tree, called {\em roots-tree ($RT$)}, which essentially keeps track of the updates on Stack $S$. 
We will map a node $N=(I,f_I)$ in roots-tree to the element $(I,f_I)$ in stack $S$.
Each {\em push} will create a new node in $RT$. The nodes are {\em never} deleted from $RT$.

\smallskip\noindent\textbf{Construction of roots-tree ($RT$):} 
Denote the root of $RT$ by $N_{\langle 0\rangle} :=(\langle 0\rangle, f_{\langle 0\rangle}:=f(x))$. Add a child node $N_{I_{0}}$ to the root corresponding to the initialization of Stack $S$ by $(I_0, f_{I_0})$, 
where $I_0:=\langle h_0(x_0)\rangle$ (label the edge $h_0$ in $RT$). 

If, at some time $t$, the algorithm pops $(I_{l-1},f_{I_{l-1}})$ from $S$ then the \emph{current node} in $RT$ will be the leaf node $N_{I_{l-1}}=(I_{l-1}, f_{I_{l-1}})$.
We map the updates on stack $S$ to $RT$ as follows:

(Step $18$) 
If ideal $I_{l-1}$ {\em grows} to $I_l:= I_{l-1}+\langle h_l\rangle$ and $(I_{l},f_{I_l})$ is pushed in $S$, 
then create a child of $N_{I_{l-1}}$ in $RT$ using an edge labelled $h_l$ (label the node $N_{I_l}:=(I_l,f_{I_l})$).  

(Steps $7, 16$) 
 If the algorithm reached {\em dead-end} (no update in stack $S$ or list $\mathcal{L}$), then add a child labelled $\mathcal{D}$ to node $N_{I_{l-1}}$. It indicates a dead-end at the current branch. 
Analogously, if the algorithm finds a {\em maximal split ideal}, we add a child labelled $\mathcal{M}$ to Node $N_{I_{l-1}}$ (indicating $I_{l-1}$ is a maximal split ideal). 

 (Step $20$) 
Suppose, processing of length-$l$ split ideal $I_{l-1}$ results in factoring each ideal $U$ in $S$, containing $h_i$, to $m$ split ideals. 
We describe the \emph{duplication process} for a particular $U$ (repeat it for each split ideal containing $h_i$).

Let $U_{i-1}$ be the length-$i$ restriction of $U$.
First, we move to the ancestor node $N_{U_{i-1}} := (U_{i-1},f_{U_{i-1}})$ of $N_{U}$. Make $m$ copies of the sub-tree at Node $N_{U_{i-1}}$, each of them attached to $N_{U_{i-1}}$ by edges labelled with $h_{i,1},\ldots,h_{i,m}$ respectively. 
The copy of each old node $N=(V,f_V)$, in sub-tree corresponding to $h_{i,j}$, will be relabelled with $(V_j,f_{V_j})$ corresponding to the factor split ideal $V_j$ of $V$ and the newly computed $f_{V_j}$. 

This step does not increase the height of the tree, though it increases the size.

\smallskip
For the rest of this section, $RT$ denotes the final roots-tree created at the end of the above process. 
We state some easy properties of $RT$, which will help us in analyzing the time complexity.

 1) By construction, size of the roots-tree increases at every iteration. We never delete a node or an edge (though relabelling might be done). 
So, the size of $RT$ bounds the number of iterations taken by Algorithm~\ref{algo1}.

 2) 
Consider a node $N_I=:(I, f_I)$ in $RT$. Here $f_I(\bar{x}_l,x)\in R[\bar{x}_l,x]$, and let $g_I\in R[\bar{x}_l,x]$ be defined as in Algorithm~\ref{algo1}, 
$g_I :=f_I(\bar{x}_l,x)/p^\alpha \bmod \hat{I}$, where $p^\alpha \;||\; f_I \mod \hat{I}$, and $\hat{I}$ is a lift of $I$ over $R$. 
Then, $g_I \bmod I$ is a nonzero polynomial over $\F_p$. 

 3) For each node $N_I=: (I, f_I(\bar{x}_{l},x))$ and its child $N_J=: (J,f_J(\bar{x}_{l+1},x))$, we have the relation, 
$ f_J \,=\, f_I(\bar{x}_{l},x_{l+1}+p x) \bmod \hat{J} $.


\smallskip
\noindent\textbf{Bounding $|RT|$:}
To bound the size of $RT$, we define a parameter for a node $N$ of $RT$, called the {\em degree} of the node $N$ and denoted by $[N]$.

\begin{definition}[Degree of a node in $RT$]
\label{def-deg-node}
The degree of root node $N_{\langle 0\rangle}$ is $[N_{\langle 0\rangle}] :=d$ ($=\deg(f)$). 
Degree of leaves $\mathcal{D}$ resp.~$\mathcal{M}$ is defined to be $1$.
	
Let $N_I=:(I, f_I)$ be a node corresponding to a split ideal $I \subseteq \F_p[\bar{x}_{l}]$, where $f_I(\bar{x}_l,x)$ belongs to $R[\bar{x}_l,x]$.
Let $p^\alpha \;||\; f_{I} \bmod \hat{I}$ and $g_I(\bar{x}_l,x) := f_{I}/p^\alpha \bmod \hat{I}$. Except, $g_I :=0$ if $\alpha\ge k$.

Then, the degree of $N$ is defined as, 
$[N] \;:=\; \max\left(1, \,\deg_{x}(g_{I} \bmod I) \times \deg(I) \right)$. 
\end{definition}

We show that the degree of a parent node bounds the sum of the degree of its children.


\begin{lemma}[Degree distributes in $RT$]
\label{lemma-property-of-degree}
Let $N$ be a node in roots-tree $RT$ and $\des(N)$ denote the set of all children of $N$. Then, $[N] \geq \sum_{C \in \des(N)} [C] $.

So, the sum of the degrees of all nodes, at any level $l$, is at least the sum of the degrees of all nodes at level $l+1$.
\end{lemma}
\begin{proof}
Let $N=(I, f_{I})$, where $I=\langle h_0,\ldots,h_l\rangle$ and $f_I(\bar{x}_l,x)\in R[\bar{x}_l,x]$. Define $\tilde{g}_{I}\in \F_p[\bar{x}_l, x]$ as 
$\tilde{g}_{I}:= g_{I}(\bar{x}_l,x) \bmod I$. 
Assume $\alpha<k$, otherwise we are done.
So, $g_I \bmod I$ is nontrivial wrt $x$; by Step $9$ (failure) and Claim~\ref{claim-monic-degree}, we get,
\begin{equation}
\label{eq-1}
\forall \bar{a}\in \calZ_{\F_p}(I): \ \deg_{x}(\tilde{g}_I \bmod I) \;=\; \deg_{x}(\tilde{g}_I(\bar{a},x)) \;.
\end{equation}

Recall $h_{l+1}(\bar{x}_l,x) := \gcd(\tilde{g}_{I}(\bar{x}_l,x), x^p-x)$.
Let $C$ be a child node of $N$ in $RT$ such that $C=:(J_C, f_{J_C})$, where $J_C=:I+\langle h_{l,C}(\bar{x}_{l+1})\rangle$ and $f_{J_C}(\bar{x}_{l+1},x) := f_I(\bar{x}_l, x_{l+1}+px) \bmod \hat{J}_C$. This gives us the factorization  
$h_{l+1}(\bar{x}_l,x) \;=\; \prod_{C\in\des(N)} h_{l,C} (\bar{x}_l,x) \bmod I$ (Step $20$, and `duplication step' when we constructed $RT$).
Again,
\begin{equation}
\label{eq-2}
	\forall \bar{b}\in \calZ_{\F_p}(J_C): \ \deg_{x}(\tilde{g}_{J_C} \bmod J_C) \;=\; \deg_{x}(\tilde{g}_{J_C}(\bar{b},x)) \;.
\end{equation}
If $g_{J_C} =: f_{J_C}/p^{v'} \bmod \hat{J_C}$ for some $v'\in \N$, by property $3$ of $RT$, we have $g_{J_C} = f_I(\bar{x}_l,x_{l+1}+p x)/p^{v'} \bmod \hat{J_C}$.
	
By definition, $[N] = \deg(I)\cdot \deg_{x}(\tilde{g}_I)$ and $[C] = \deg(J_C)\cdot \deg_{x}(\tilde{g}_{J_C})$. 
Since $\deg(J_C) = \deg(I)\cdot \deg_x(h_{l,C}(\bar{x}_{l}, x))$, the lemma statement is equivalent to showing,
\begin{equation}
\label{eq-dist-root}
	\deg_{x}(\tilde{g}_I) \,\ge\, \sum_{C \in \des(N)}  \deg_x(h_{l,C}(\bar{x}_{l}, x))\cdot \deg_{x}(\tilde{g}_{J_C}) \;.
\end{equation}

Continuing with the notation of a particular child $C$, fix an $\bar{a}\in \calZ_{\F_p}(I)$. 
Since $J_C$ is a split ideal, $h_{l,C}(\bar{a},x)$ (of degree $d'_C$) can be written as $\prod_{i=1}^{d'_C} (x-c_i)$, where each $c_i\in \F_p$ and are distinct. 
Then, each $c_i$ is also a root of $\tilde{g}_I(\bar{a},x)$, say with multiplicity $m_i\in \N$. 
So, there exists $G(x) \in \F_p[x]$ (coprime to $x-c_i$), such that, 
	$\tilde{g}_I(\bar{a},x) \;\equiv\; (x-c_i)^{m_i}\cdot G(x) \bmod p $. 
Lifting this equation mod $p^k$, there exists $G_1(x)\in R[x]$, of degree less than $m_i$, and a unique lift $G_2(x)\in R[x]$ of $G(x)$ (Hensel lemma (\ref{hensel-lemma})) : 
	$g_I(\bar{a},x) \;\equiv\; ((x-c_i)^{m_i}+ p G_1(x))\cdot G_2(x) \bmod p^k$ . 
Substituting $x \to c_i+px$, we get, 
	$g_I(\bar{a},c_i+p x) \;\equiv\; ((p x)^{m_i}+ p G_1(c_i+p x))\cdot G_2(c_i+p x) \bmod p^k$ . 

Let $\bar{b}_i=(\bar{a},c_i)\in \calZ_{\F_p}(J_C)$. We know that $\tilde{g}_{J_C}(\bar{b}_i,x) = f_I(\bar{a},c_i+p x)/p^{v'} \bmod p$ is nontrivial. This implies that, 
$((p x)^{m_i}+ p G_1(c_i+p x))/p^{v'} \bmod p$ is a nonzero polynomial of degree at most $m_i$ ($\because p\nmid G_2(c_i)$).

Since $G_2(c_i+p x)\not\equiv 0\bmod p$ is a unit, $\deg_{x}(\tilde{g}_{J_C}(\bar{b}_i,x)) = \deg_{x}(\tilde{g}_{J_C})\leq m_i$ (Eqn.~\ref{eq-2}). 
Summing up over all the roots $c_i$ of $\tilde{g}_I(\bar{a},x)$,	
	\[ \sum_{i=1}^{d'_C} \deg_{x}(\tilde{g}_{J_C}(\bar{b}_i,x)) \;=\; d'_C\cdot \deg_{x}(\tilde{g}_{J_C}) \;\le\; \sum_{i=1}^{d'_C} m_i  \;=: d_C(g_I)\;.\]	
Summing over all children $C\in\des(N)$ (using Eqn.~\ref{eq-1}, factorization of $h_{l+1}$ \& distinctness of $\F_p$-roots), we deduce,
	\[  \sum_{C\in\des(N)} \deg_x(h_{l,C}) \deg_{x}(\tilde{g}_{J_C}) \;\leq\; \sum_C d_C(g_I)  \;\leq\; \deg_{x}(\tilde{g}_I(\bar{a},x)) \;=\; \deg_{x}(\tilde{g}_I) \;.\]
This proves Eqn.~\ref{eq-dist-root}, and hence the lemma.
\end{proof}
%
%
%
%
%

Define the degree of list $\mathcal{L}$ as, $\deg(\mathcal{L}):= \Sigma_{I\in \mathcal{L}} \deg(I)$. 

\begin{lemma}[Bounding $|RT|$, $\deg(I)$, $\deg(\mathcal{L})$, $|\mathcal{L}|$]
\label{lemma-iterations}
Let $RT$ be the roots-tree constructed from the execution of Algorithm~\ref{algo1}.
The number of leaves of $RT$, resp.~$\deg(\mathcal{L})$, is at most $d=\deg(f(x))$. Also, the size $|RT|$ of the roots-tree (hence, the number of iterations by Algorithm \ref{algo1}) is bounded by $dk$.
\end{lemma}
\begin{proof}
Applying Lemma~\ref{lemma-property-of-degree} inductively, sum of the degrees of nodes at any level is bounded by the degree $d$ of the root node.
In particular, 

 1) We can extend every leaf to bring it to the last level (create a chain of nodes of same degree) without changing the degree distribution property. So, $\deg(\mathcal{L})= \Sigma_{I\in \mathcal{L}} \deg(I) \leq d$. 
	      Since the number of leaves is $\ge|\mathcal{L}|$, we get $|\mathcal{L}| \leq d$. 

2) For any split ideal $I$ in stack $S$, $\deg{I} \leq d$.
	
3) Since the depth of the roots-tree is at most $k$, $|RT| \leq kd$.
\end{proof}


\begin{lemma}[Computation cost at a node]
\label{lemma-comp-cost}
Computation cost at each node of $RT$ (time taken by Algorithm~\ref{algo1} in every iteration of the while loop) is bounded by $\poly(d, k\log p)$.
\end{lemma}
\begin{proof}
During an iteration, the major computations performed by the algorithm are--- 
testing for zerodivisors (Step $9$), computing modular gcd (Step $13$), computing reduced $f_I$ (Steps $3,18$), performing reduction for repeated squaring (Step $12$), and factoring ideals (Step $20$).

These operations are described by Lemmas~\ref{lemma-reduction}, \ref{lem-div-mod-I}, \ref{lemma-stack-update}, \ref{lemma-zero-div} and \ref{lemma-gcd}.
All of them take time $\poly(d, k\log p$, $\deg(I))$, where $I$ is the concerned triangular ideal. 

For any split ideal $I$ (or its lift $\hat{I}$), we know that $\deg(I)\leq d$ (Lemma~\ref{lemma-iterations}). So, Steps $3, 9, 13, 18, 20$ take time $\poly(d, k\log p)$. 
Step $12$ to compute repeated squaring modulo $I+\langle\tilde{g}\rangle$ takes time $\poly(\deg_{x}(\tilde{g}), \deg(I), k\log p)$ (using Lemma~\ref{lemma-reduction}). 
Since $I$ is a split ideal with $\deg(I)\leq d$, and degree of $\tilde{g}$ is at most $d$, so Step $12$ also takes $\poly(d, k\log p)$ time.

Hence the computation cost at each node is $\poly(d, k\log p)$.
\end{proof}

\begin{proof}[Proof of Theorem \ref{thm1}]
The definition of roots-tree shows that the number of leaves upper bound the number of all maximal split ideals in $\mathcal{L}$.
Lemmas~\ref{lemma-iterations} and \ref{lemma-comp-cost} show that the time complexity of Algorithm~\ref{algo1} is bounded by $\poly(d, k\log p)$ 
(by bounding both number of iterations and the cost of computation at each iteration). 
Using Lemma~\ref{lemma-MSI-root-represent} on the output of Algorithm~\ref{algo1}, we get the exact count on the number of roots of $f\bmod p^k$ in time $\poly(d, k\log p)$.
\end{proof}

\section{Proof of Theorem \ref{thm2}} \label{sec-main2}
\vspace{-1mm}

A polynomial $f$ can be factored mod $p^k$ if it has two basic-irreducible factors of different degree (using distinct degree factorization~\cite{von2001factoring} and Hensel Lemma~\ref{hensel-lemma}).   

If two basic-irreducible factors appear with different exponents/multiplicities, then again $f$ can be factored (using formal derivatives~\cite{von2001factoring} and Hensel Lemma~\ref{hensel-lemma}).

So, for factoring $f\bmod p^k$, we can assume $f\equiv (\varphi_1\ldots \varphi_t)^e + p h \bmod p^k$, 
where every $\varphi_i \in (\Z/\langle p^k\rangle)[x]$ is a basic-irreducible polynomial of a fixed degree $b$. Also, $d:=\deg(f)= bte$.
Let us fix this assumption for this section, unless stated explicitly. 

A basic-irreducible factor of $f\bmod p^k$ has the form $\varphi_i + p w_i(x) \bmod p^k$, for $i\in [t]$ (Lemma~\ref{lemma-f-is-power}).

If $b=1$, counting basic-irreducible factors of $f$ is equivalent to counting roots of $f$. 

When $b>1$, we prove a simple generalization of this idea; it is enough to count all the roots of $f$ in the ring extension $\Z[y]/\langle p^k, \varphi(y)\rangle$, 
where $\varphi(y)$ is an irreducible mod $p$ of degree-$b$. These rings are called {\em Galois rings}, we denote them by $G(p^k, b)$ (unique, for fixed $k$ and $b$, up to isomorphism).


\noindent
\subsection{Reduction to root-counting in $G(p^k,b)$}\label{sec-reduction-root}

By Lemma \ref{lemma-f-is-power}, any basic-irreducible factor of $f \bmod p^k$ is a factor of a unique $({\varphi_i}^e + p w_i(x))$; and $\varphi_i$ are coprime mod $p$. 
So in this subsection, for simplicity of exposition, we will assume that $f(x)$ equals $\varphi^e \bmod p$ ($\varphi$ is a monic degree-$b$ irreducible mod $p$).

Define $G := G(p^k,b)$. Let $y_0,y_1,\ldots,y_{b-1}$ be the roots of $\varphi(x)$ in $G$ (Claim~\ref{claim1}). Wlog, taking $y:=y_0$, $y_i \equiv y^{p^i} \bmod p$, for all $i \in \{0,\ldots,b-1 \}$ (Frobenius conjugates in $\F_p$).
Note that $G \cong (\Z/\langle p^k \rangle)[y]=:G'$. We will prefer to use $G'$ below.

The lemma below associates a root of $f$, in $G$ or $G'$, to a unique basic-irreducible factor of $f$ in $(\Z/\langle p^k\rangle) [x]$.

\begin{lemma}[Root to factor]
\label{lemma-root-to-bif}
Let $r(y)\in G'$ be a root of $f(x)$. Then, $h(x):=\prod_{i=0}^{b-1} (x-r(y_i))$ is the unique basic-irreducible factor of $f$ having root $r(y)$. 
We say: $h(x)$ is the basic-irreducible factor {\em associated to root} $r(y)$.
\end{lemma}

\begin{proof}
The coefficients of $h$ are symmetric polynomials in $r(y_i)$ (over $0\le i<b$). Since the automorphism $\psi_1:y \to y_1$ of $G'$ (as defined in Claim \ref{claim2}) permutes $r(y_i)$'s ($\because$ it permutes $y_i$'s), 
it fixes all the coefficients of $h$. From Claim~\ref{claim2}, all these coefficients are then in $\Z/\langle p^k \rangle$. Hence, $h \in (\Z/\langle p^k\rangle) [x]$.

If $r(y)$ is a root of another polynomial $h'$ in $(\Z/\langle p^k \rangle) [x]$, then $r(y_i)$'s are also roots of $h'$ (applying automorphisms $\psi_i$ of $G'$). Since these roots are coprime mod $p$, we actually get: $h|h'$. Thus, $h$ is the unique monic irreducible factor of $f$ containing $r(y)$.

Looking mod $p$, $r(y_i)$'s are a permutation of the roots of $\varphi(x)$, so $h(x) \equiv \varphi(x) \bmod p$. Hence, $h(x)$ is the unique monic basic-irreducible factor of $f$.
\end{proof}


Following is the reduction to counting all roots of $f$ in $G$.

\begin{theorem}[Factor to root]
\label{theorem-reduction}
Any degree-$b$ basic-irreducible factor of $f\bmod p^k$ has exactly $b$ roots in $G$. Conversely, if $f$ has a root $r(y)\in G$, then it must be a root of a unique degree-$b$ basic-irreducible factor of $f\bmod p^k$.

So, the number of degree-$b$ basic-irreducible factors of $f\bmod p^k$ is exactly the number of roots, of $f$ in $G$, divided by the degree $b$.
\end{theorem}

\begin{proof}
By Lemma \ref{lemma-root-to-bif} (\& uniqueness of Galois rings), for every root $r(y)\in G$ of $f$, we can associate a unique basic-irreducible factor of $f(x)$.

Conversely, let $h(x)=:\varphi(x)+p w(x)$ be a basic-irreducible factor of $f(x)$. It splits completely in $G$ (as, $h(x)\equiv \varphi \bmod p$; first factor in $G/\langle p\rangle$ and then Hensel lift to $G$).
So, $h$ has exactly $b$ roots in $G$, each of them is also a root of $f$ in $G$. 

Hence the theorem statement follows.
\end{proof}

{\bf Remark.} This `irreducible factor vs root' correspondence, for $f\bmod p^k$,  breaks down if $G$ is {\em not} a Galois ring. Eg. for the ring $\Z[y]/\langle p^k, y^2-p\rangle$?

\subsection{Counting roots in $G(p^k, b)$-- Wrapping up Thm.~\ref{thm2}}\label{sec-count-root}

In this section, we show how to count the roots of $f\equiv (\varphi_1 \varphi_2 \ldots \varphi_t)^e+p h(x) \bmod p^k$ in $G(p^k, b)$. 
Since $G := G(p^k, b)$ is a Galois ring, so $G/\langle p\rangle= \F_{p^b} =: \F_q$.
(Recall: $R=\Z/\langle p^k\rangle$.)

\smallskip
\textbf{Split ideals and zerosets in the Galois ring:} 
First, we will modify the definition of zerosets (Section~\ref{sec-pre}) to include zeros of $f$ in $G$. 
A {\em $G$-zeroset} of $f(x)\in R[x]$ will be defined as $\calZ_{G}(f):=\{ r\in G \mid f(r)\equiv 0 \bmod p^k \}$.
Similarly, for an ideal $I\subseteq \F_p[\bar{x}_l]$, its $\F_q$-zeroset is defined as $\calZ_{\F_q}(I):=\{ \bar{a}=(a_0,\ldots,a_l)\in (\F_q)^{l+1} \mid g(\bar{a})\equiv 0 \bmod p^k, \forall g \in I \}$.

The definition of triangular ideals, split ideals and maximal split ideals will remain exactly same (generators defined over $\F_p$, Section \ref{sec-pre}),
except that in the third condition for split ideals, zeroset will be over $\F_q$ instead of $\F_p$. But, they can now be seen as 
storing potential roots of $f(x)$ in $G$ (or, storing potential basic irreducible factors of $f\bmod p^k$).
The reason is, a root $r(y)\in G$ of $f\bmod p^k$ can be viewed as, $r(y)=r_0(y)+ p r_1(y)+ p^2 r_2(y)+\ldots+ p^{k-1} r_{k-1}(y)$, where each $r_i(y)\in G/\langle p\rangle=\F_q$.
So, the decomposition of formal variable $x=:x_0+p x_1+ p^2 x_2+\ldots+ p^{k-1} x_{k-1}$, now represents candidates for $r_0$, $r_1$, and so on, over $\F_q$.

A split ideal $I_l \subseteq \F_p[\bar{x}_l]$, defined as $I_l:=\langle h_0(x_0),\ldots, h_l(\bar{x}_l) \rangle$, now implicitly stores the candidates for $(r_0)$ in $h_0$, $(r_0, r_1)$ in $h_1$, and so on. These, in turn, give candidates for basic-irreducible factors of $f\bmod p^{l'}$ (some $l'\le k$).

In particular, when $I_l$ is a maximal split ideal, an $\bar{r}_l$ implicitly denote a basic-irreducible factor of $f\bmod p^{k}$. The number of such factors is $\deg(I_l)\cdot p^{k-l-1}/ b$ (Theorem~\ref{theorem-reduction} \& Lemma~\ref{lemma-MSI-root-represent}). 

Split ideals follow all the properties given in Section \ref{sec-pre}, just by replacing the fact that roots belong to $\F_q$ and not $\F_p$. 

\smallskip
\textbf{Description of the modified algorithm:} Algorithm~\ref{algo1}, to count roots in $R$, extends directly to count roots in $G$. 
The algorithm is exactly same except one change: to compute GCD (Steps $3$ and $13$), we now use the Frobenius polynomial $x^q-x$ instead of the prior $x^p-x$ 
(GCD computation implicitly stores the candidate roots, they are in $\F_q$ now).

So the algorithm works as follows:
\begin{enumerate}

\item It gets $f(x)\equiv (\varphi_1 \ldots \varphi_t)^e + p w(x) \bmod p^k$ as input, computes gcd $h_0(x):= \gcd(f(x), x^q-x)$ over $\F_p$. 
Since $x^q-x$, over $\F_p$, is the product of all irreducible factors of degree dividing $b$, we deduce: $h_0(x)= \varphi_1 \ldots \varphi_t \bmod p$; and define the first split ideal $I_0 :=\langle h_0 \rangle$. (Note-- We do not have access to $\varphi_i$'s themselves.)

\noindent{\bf Remark.}
The length $1$ split ideal stores all the roots of $f$ in $G/\langle p\rangle$, or all the basic irreducible factors of $f\bmod p$; as $h_0(x)= \varphi_1 \ldots \varphi_t$. 
Also, its degree is $t b$, which when divided by $b$, gives the count of the basic-irreducible factors of $f\bmod p$.

\item The algorithm then successively looks for the next precision candidates. It computes $h_l$ by taking gcd with $x^q-x$, and adds it to the previous ideal $I_{l-1}$ like before.

\item All the supporting algebraic algorithms and lemmas (given in appendix) work the same as before; since they are being passed the same parameters--- a split ideal, or a triangular ideal, or a polynomial over $R$.

\end{enumerate}

Thus, a similar proof of correctness and time complexity can be given as before. 

\begin{proof}[Proof of Theorem \ref{thm2}]
Consider a univariate $f(x)\bmod p^k$. As discussed in the beginning of this section, $f\bmod p^k$ can be efficiently factorized as $f\equiv \prod_{i=1}^{m} f_i \bmod p^k$, 
where each $f_i(x)$ is a power of a product of degree-$b_i$ irreducible polynomials mod $p$ (i.e.~of the form $\equiv (\varphi_1 \varphi_2 \ldots \varphi_t)^e+p h(x)$, where $\varphi_j$ is a degree-$b_i$ irreducible mod $p$).

On each such $f_i \bmod p^k$, we use Algorithm~\ref{algo1} with the new Frobenius polynomial $(x^{q_i}-x)$ ($q_i = p^{b_i}$), in Steps $3$ and $18$, as discussed above. 
Let the final list output, for $f_i\bmod p^k$, be $\mathcal{L}_i =: \{I_1(l_1, D_1),\ldots,I_n(l_n,D_n)\}$. 
Thus, we get the count on the $G(p^k,b_i)$-roots of $f_i\bmod p^k$ as $\Sigma_{j=1}^{n} D_j q_i^{k-l_j}$ (Lemma~\ref{lemma-MSI-root-represent}). 
Using Theorem \ref{theorem-reduction}, the number of the degree-$b_i$ basic-irreducible factors of $f\bmod p^k$ is $B_k(f_i):=$ $(1/b_i)\times\Sigma_{j=1}^{n} D_j q_i^{k-l_j}$.

Using Lemma \ref{lemma-f-is-power}, we get the count on the basic-irreducible factors of $f\bmod p^k$ as, $B_k(f)= \Sigma_{i=1}^{m} B_k(f_i)$. 

For the time complexity, only difference is the repeated-squaring to compute the reduced form of polynomial $x^{q_i}-x$ (Steps $3, 12$), it will take $b_i \log p$ operations instead of $\log p$ operations.
But $b_i \leq d$, so the algorithm runs in time $\poly(d, k\log p)$ (\& remains deterministic).

\end{proof}

\vspace{-1mm}
\section{Conclusion}
\vspace{-1mm}

There are well known efficient deterministic algorithms to count the number of roots/irreducible factors over prime characteristic. 
Surprisingly, not many results are known when the characteristic is a {\em prime-power}. The main difficulty is that the ring has {\em non}-unique factorization.


We give the first efficient deterministic algorithm to count the number of basic-irreducible factors modulo a prime-power. Restricting it to degree-one irreducibles, we get a deterministic polynomial-time algorithm to count the roots too. This is achieved by storing and improving roots (wrt precision) virtually using \emph{split ideals} (we do not have access to roots directly).
As a corollary: we can compute the Igusa zeta function deterministically, and we also get a deterministic algorithm to count roots in $p$-adic rings (resp.~formal power-series ring). 

Many interesting questions still remain to be tackled. For $p$-adic fields, there is only a randomized method to count the number of irreducible factors.
Analogously, the question of counting irreducible factors modulo a prime-power also remains open; no efficient method is known even in the randomized setting. The {\em ramified} roots seem to elude practical methods. 
On the other hand, the problem of actually {\em finding} an irreducible factor (resp.~a root) deterministically, seems much harder; it subsumes the analogous classic problem in prime characteristic.

\smallskip
\noindent
{\bf Acknowledgements. }
We thank Vishwas Bhargava for introducing us to the open problem of factoring $f \bmod p^3$ and the related prime-power questions. A.D.~thanks Sumanta Ghosh for the discussions.
N.S.~thanks the funding support from DST (DST/SJF/MSA-01/2013-14).
R.M.~would like to thank support from DST through grant DST/INSPIRE/04/2014/001799.


\bibliographystyle{alpha}
\bibliography{bibliography}

\appendix

\section{Preliminaries}
\label{appen-prelim}


{\bf Lifting factorization:} Below we state a lemma, originally due to Kurt Hensel~\cite{Hensel1918}, for $\mathcal{I}$-adic lifting of factorization of a given univariate polynomial. 
Over the years, Hensel's lemma has acquired many forms in different texts, version presented here is due to Zassenhaus~\cite{zassenhaus1969hensel}.

\begin{lemma}[Hensel's lemma \cite{Hensel1918}]
\label{hensel-lemma}
Let $R$ be a commutative ring with unity, denote the polynomial ring over it by $R[x]$. 
Let $\mathcal{I}\subseteq R$ be an ideal of ring $R$. Given a polynomial $f(x) \in R[x]$, suppose $f$ factorizes as 
\[f=gh \bmod \mathcal{I} , \] 
such that $gu+hv=1 \bmod  \mathcal{I}$ (for some $g,h,u,v \in R[x]$). 
Then, given any $l \in \mathbb{N}$, we can efficiently compute $g^*,h^*,u^*,v^* \in R[x]$, such that,
\[f=g^* h^* \bmod \mathcal{I}^{l} .\]
Here $g^*=g \bmod \mathcal{I}$, $h^*=h \bmod \mathcal{I}$ and $g^*u^*+h^*v^*=1 \bmod \mathcal{I}^{l}$ (i.e.~{\em pseudo}-coprime lifts). Moreover $g^*$ and $h^*$ are unique up to multiplication by a unit.
\end{lemma}

Using Hensel's lemma, for the purpose of counting roots (resp.~basic-irreducible factors), a univariate polynomial $f(x)\in \mathbb{Z}[x]$ can be assumed to be a power of an irreducible modulo $p$.

\begin{lemma}
\label{lemma-f-is-power}
By the fundamental theorem of algebra, a univariate $f(x)\in \mathbb{Z}[x]$ factors uniquely, over $\mathbb{F}_p$, into coprime powers as,  
$f \equiv \prod_{i=1}^{m} {\varphi_i}^{e_i}$ ,
where each $\varphi_i \in \mathbb{Z}[x]$ is irreducible mod $p$ and $m, e_i \in \mathbb{N}$. Then, for all $k\in\N$,
\begin{enumerate}
\item $f$ factorizes mod $p^k$ as $f=g_1 g_2 \ldots g_m$, where $g_i$'s are mutually co-prime mod $p^k$ and $g_i\equiv {\varphi_i}^{e_i} \bmod p$, for all $i\in [m]$.

\item any basic-irreducible factor of $f(x) \bmod p^k$ is a basic-irreducible factor of a unique $g_j \bmod p^k$, for some $j\in [m]$. 
Let $B_k(h)$ denote the number of (coprime) basic-irreducible factors of $h(x) \bmod p^k$. Then,
$ B_k(f)= \Sigma_{i=1}^{m} B_k(g_i)$ .

\item any root of $f$ mod $p^k$ is a root of a unique $g_i$ mod $p^k$. Let $N_k(h)$ denote the number of (distinct) roots of $h(x) \bmod p^k$. Then,
$ N_k(f)= \Sigma_{i=1}^{m} N_k(g_i) $.
\end{enumerate}
\end{lemma}

\begin{proof}
We can apply Hensel's lemma by taking ring $R:=\mathbb{Z}$ and ideal $\mathcal{I}:=\langle p\rangle$. The co-prime factorization of $f \bmod p$ lifts to a 
unique coprime factorization $f\equiv g_1 g_2\ldots g_m \bmod p^k$, for any $k \in \mathbb{N}$ and $g_i \equiv {\varphi_i}^{e_i} \bmod p$.

Any basic-irreducible factor $h(x)$ of $f(x) \bmod p^k$ has to be $h\equiv \varphi_i \bmod p$ for some $i \in [m]$; otherwise, $h$ will become reducible mod $p$. 
Since $g_i$'s are co-prime and $h|f \bmod p^k$, $h$ must divide a unique $g_i$. So, any basic-irreducible factor $h$ of $f(x) \bmod p^k$ is a basic-irreducible factor of a unique $g_j \bmod p^k$.
Clearly, any basic-irreducible factor of a $g_i$ is also a basic-irreducible factor of $f \bmod p^k$. This proves $B_k(f)= \Sigma_{i=1}^{m} B_k(g_i)$.

The third part follows from a similar reasoning as the second part.
\end{proof}


{\bf Root finding over a finite field:} The following theorem, called CZ in this paper and given by Cantor-Zassenhaus~\cite{cantor1981new}, finds all roots of a given univariate polynomial over a finite field in randomized polynomial time. (Equivalently, it finds all irreducible factors as well.)

\begin{theorem}[Cantor-Zassenhaus Algo (CZ)]
\label{thm-CZ}
Given a univariate degree $d$ polynomial $f(x)$ over a finite field $\mathbb{F}_q$, all roots of $f$ in $\mathbb{F}_q$ can be found in randomized poly($d,\log{q}$) time.
\end{theorem}

\subsection{Properties of Galois rings-- Analogues of finite fields}

A {\em Galois ring}, of characteristic $p^k$ and size $p^{kb}$, is denoted by $G(p^k,b)$ (where $p$ is a prime, $k,b \in \mathbb{N}$). 
It is known that two Galois rings of same characteristic and size are isomorphic to each other.
We will define Galois ring $G(p^k,b)$ as the ring $\G :=\mathbb{Z}[y]/\langle p^k, \varphi(y)\rangle$, where $\varphi(y)\in \mathbb{Z}[y]$ is an irreducible mod $p$ of degree $b$ \cite{mcdonald1974finite}. 
Let us prove some useful properties of $\G$ below.

\begin{claim}[Roots of $\varphi$]
\label{claim1}
Let $\varphi'(x)\in \mathbb{Z}[x]$ be any irreducible mod $p$ of degree $b$.
There are $b$ distinct roots of $\varphi'(x)$ in $\G$.
Let $r$ denote one of the roots, then all other roots, modulo $p$, are of the form $r^{p^i}$ ($i \in \{0,\ldots,b-1 \}$). 
\end{claim}
\begin{proof}
$\G/\langle p\rangle$ is isomorphic to the finite field of degree $b$ over $\mathbb{F}_p$.
So, irreducible $\varphi'(x)\in \mathbb{F}_p[x]$ has exactly $b$ roots in $\G/\langle p\rangle$ \cite[Ch.2]{lidl1994introduction}. 
By Hensel Lemma \ref{hensel-lemma}, roots in $\G/\langle p\rangle$ can be lifted to $\G$ uniquely. 
Hence, $\varphi'(x)$ has exactly $b$ distinct roots in $\G$. Modulo $p$, they are of the form $r^{p^i}$ ($i \in \{0,\ldots,b-1\}$) for a root $r$ (lifted from roots in $\G/\langle p \rangle$).
\end{proof}

Using Claim~\ref{claim1}, denote roots of $\varphi(x)$ as $y_0,\ldots,y_{b-1}$; here $y_i \equiv y_0^{p^i} \bmod p$ for all $i \in \{0,\ldots,b-1 \}$. 
For all roots $y_j$, $\G \equiv R[y_j]$. In other words, $y_j$ generate the extension $\G$ over $R$. 

\begin{claim}[Symmetries of $\G$]
\label{claim2}
There are exactly $b$ automorphisms of $\G$ fixing $R = \mathbb{Z}/ \langle p^k \rangle$, denoted by $\psi_j$ ($j\in \{0,\ldots,b-1\}$). 
Each of these automorphisms can be described by a map taking $y_0$ to one of the roots of $\varphi(x)$ and fixing $R$.
Wlog, assume $\psi_j$ maps $y_0 \to y_j$. 
	
Moreover, for all $j$ coprime to $b$, $\psi_j$ fixes $R$ and nothing else.
\end{claim}
\begin{proof}

Since coefficients of $\varphi(x)$ belong to $R$, an automorphism fixing $R$ should map the root $y_0=:y$ to another of its roots $y_j$. 
We only need to show that $\psi_j$ is an automorphism (it is a valid map because $y_j \in \G$) 

Writing elements of $\G$ in terms of $y_0$ (i.e.~$\G \cong R[y_0]$), it can be verified that $\psi_j(ab)=\psi_j(a)\psi_j(b)$ and $\psi_j(a+b)=\psi_j(a)+\psi_j(b)$, so $\psi_j$ is a homomorphism.

Similarly, if $\psi_j(g) = 0$, writing $g$ in terms of $y_0$, we get that $g=0$. So, kernel of $\psi_j$ is the set $\{0\}$; thus, it is an isomorphism.

\smallskip
For the moreover part, let $\psi_j$ be such that $j$ is coprime to $b$. 
We will show a stronger statement by induction: for any $i\leq k-1$, if $a(y_0)=\psi_j (a(y_0))$ in $\G/\langle p^i \rangle$, then $a(y_0) \in \mathbb{Z}/\langle p^i\rangle$.

{\em Base case:} If $i=1$ and $j=1$, then $a(y_0)=\psi_1(a(y_0)) \bmod p \Rightarrow a(y_0)={a(y_0)}^p \bmod p$. It means $a(y_0) \in \mathbb{Z}/\langle p\rangle$.
 
If $j$ is coprime to $b$, then $\psi_j$ generates $\psi_1$ modulo $p$. So, $a(y_0) = \psi_j(a(y_j)) \bmod p$ implies that, $a(y_0) \bmod p =: a_0 \in \mathbb{Z}/\langle p \rangle$.

This argument also proves: for any $i\leq k$, if $a(y_0)=a(y_j)$ in $\G/\langle p^i \rangle$, then $a(y_0) \in \mathbb{F}_p$ (in other words, $a(y_0)$ is $y_0$ free).

{\em Induction step:} Let us assume that $a(y_0) = \psi_j(a(y_0))$ in $\G/\langle p^i \rangle$. 
By the previous argument, $a(y_0) = a_0 + p a'(y_0)$, where $a_0 \in \mathbb{Z}/\langle p \rangle$ and $a'(y_0) \in \G/\langle p^{i-1}\rangle $. 

From the definition, $a(y_0)= \psi_j (a(y_0))$ iff $a'(y_0)=\psi_j(a'(y_0))$ in $\G/\langle p^{i-1}\rangle$. 
By induction hypothesis, the latter is equivalent to $a'(y_0) \in \mathbb{Z}/\langle p^{i-1}\rangle$. 
So, $a(y_0) \in \mathbb{Z}/\langle p^i\rangle$.

\smallskip
Hence, the only fixed elements under the map $\psi_j$ ($j$ coprime to $b$) are integers; in $\mathbb{Z}/\langle p^k\rangle$.
\end{proof}

\section{Proofs of Section~\ref{sec-pre}}
\label{app-split-proofs}


\begin{proof}[Proof of Lemma~\ref{lemma-split-ideal-is-recursive}]
It is enough to show the lemma for $j=l-1$. It is easy to observe that $I_{l-1}$ is triangular.

Looking at the second condition for being a split ideal, $|\calZ_{\F_p}(I_{l-1})|  \leq \prod_{i=0}^{l-1} \deg_{x_i}(h_i)$ follows because a degree $d\ge1$ polynomial 
can have at most $d$ roots in $\mathbb{F}_p$. 

To show equality, notice that for any $\bar{a}=(a_0,\ldots,a_{l-1}) \in \mathcal{Z}_{\mathbb{F}_p}(I_{l-1})$, $\deg_{x_l}(h_l(\bar{a},x_l))$ is bounded by $\deg_{x_l}(h_l)$.
This implies $h_l(\bar{a},x_l)$ can have at most $\deg_{x_l}(h_l)$ roots in $\F_p$.
If $|\calZ_{\F_p}(I_{l-1})| < \prod_{i=0}^{l-1} \deg_{x_i}(h_i)$ then $|\calZ_{\F_p}(I_{l})| < \deg_{x_l}(h_l) \cdot\prod_{i=0}^{l-1} \deg_{x_i}(h_i)$, contradicting that $I_l$ is a split ideal. 
	\footnote{This argument also shows that every $\F_p$-zero of $I_{l-1}$ `extends' to exactly $\deg_{x_l}(h_l)$ many $\F_p$-zeros of $I_l$.} 

For the third condition, since $I_l$ is a split ideal, for any $(a_0,\ldots,a_{l-1}) \in \calZ_{\F_p}(I_{l-1})$,
$f(a_0+p a_1+\ldots+p^l a_l)\equiv 0 \bmod p^{l+1} \Rightarrow f(a_0+p a_1+\ldots+p^{l-1} a_{l-1})\equiv 0 \bmod p^{l}.$
\end{proof}
%
%

Lemma~\ref{lemma-split-ideal-str} shows that a split ideal $I$ can be decomposed in terms of ideals $I_{\bar{a}} := \langle x_0-a_0,\ldots,x_l-a_l \rangle$, where $\bar{a} =: (a_0,\ldots,a_l)$ is a root of $I$.
Before we prove this structural lemma, let us see some properties of these ideals $I_{\bar{a}}$'s.

\begin{claim} \label{properites_maximal_ideal}
Let $I$ be a split ideal.
\begin{enumerate}
	\item For any ideal $I_{\bar{a}}$, quotient $\F_p[x_0,\ldots,x_l]/I_{\bar{a}} \cong \F_p$ is a {\em field}. 
	\item $I_{\bar{a}}$ and $I_{\bar{b}}$ are {\em coprime} for any two distinct roots $\bar{a}, \bar{b} \in \calZ_{\F_p}(I)$. This is because there exists $i$, for which $a_i \neq b_i$; yielding $(a_i-b_i)^{-1} \left((x_i-b_i)-(x_i-a_i)\right) = 1$ in the sum-ideal $I_{\bar{a}}+I_{\bar{b}}$.
	\item $I_{\bar{a}} \cap I_{\bar{b}} = I_{\bar{a}} I_{\bar{b}}$ for any two distinct roots $\bar{a}, \bar{b} \in \calZ_{\F_p}(I)$. 
	      It follows because there exist $r_{\bar{a}} \in I_{\bar{a}}$ and $r_{\bar{b}}\in I_{\bar{b}}$, s.t., $r_{\bar{a}} + r_{\bar{b}} = 1$. 
	      So, $r \in I_{\bar{a}} \cap I_{\bar{b}} \Rightarrow r = r(r_{\bar{a}}+r_{\bar{b}}) \in I_{\bar{a}} I_{\bar{b}}$. On the other hand, $I_{\bar{a}} I_{\bar{b}} \subseteq I_{\bar{a}} \cap I_{\bar{b}}$ follows from the definition of the product-ideal.
	\item Generalizing the previous point--- for a set $A$ of distinct roots $\bar{a}$'s, $\bigcap_{\bar{a}\in A} I_{\bar{a}} = \prod_{\bar{a}\in A} I_{\bar{a}}$.
\end{enumerate}
\end{claim}

\begin{proof}[Proof of Lemma~\ref{lemma-split-ideal-str}]

We will prove this decomposition by applying induction on the length of the split ideal. For the base case, length of $I$ is $1$ and $I=\langle h_0(\bar{x}_0)\rangle \subseteq \F_p[x_0]$.
Since $I$ is a split ideal, $h_0(x_0) = \prod_{i=1}^{\deg(h_0)} (x_0 - a_i)$ for {\em distinct} $a_i \in \F_p$. 
So, $I = \prod_{i=1}^{\deg(h_0)} I_{a_i} = \bigcap_{i=1}^{\deg(h_0)} I_{a_i}$ by Claim~\ref{properites_maximal_ideal}.

\smallskip
Let $I$ be a split ideal of length $l+1$, $I=:\langle h_0(\bar{x}_0),\ldots,h_l(\bar{x}_l)\rangle \subseteq \F_p[x_0,\ldots,x_l]$. 
Define ideal $I':=\langle h_0(\bar{x}_0),\ldots,h_{l-1}(\bar{x}_{l-1}) \rangle$. By Lemma~\ref{lemma-split-ideal-is-recursive}, $I'$ is a split ideal. 
From the induction hypothesis (\& Claim~\ref{properites_maximal_ideal}), we have $I'=\bigcap_{\bar{a}\in \calZ_{\F_p}(I)} I'_{\bar{a}} = \prod_{\bar{a}} I'_{\bar{a}}$, where $I'_{\bar{a}} :=\langle x_0-a_0,\ldots,x_{l-1}-a_{l-1} \rangle$ for a zero $\bar{a} =: (a_0,\ldots,a_{l-1})$ of $I'$.
We know that, 
\begin{equation} \label{ideal_product}
	I \;=\; I' + \langle h_l(\bar{x}_l) \rangle \;=\; \prod_{\bar{a}\in \calZ_{\F_p}(I')} \left(I'_{\bar{a}} + \langle h_l(\bar{x}_l) \rangle \right) \;.
\end{equation}
Claim~\ref{claim-monic-degree} shows $\deg(h_l(\bar{a},x_l))=\deg_{x_l}(h_l)$ for all $\bar{a}\in \calZ_{\F_p}(I')$, and 
$h_l(\bar{a},x_l)$ splits completely over $\F_p$. 
So, for any $\bar{a} \in \calZ_{\F_p}(I')$, $I'_{\bar{a}}+\langle h_l(\bar{x}_l) \rangle = \prod_{i=1}^{\deg_{x_l}(h_l)} I_{\bar{a}, b_i} $, where $(\bar{a},b_i)$ are roots of $I$ extended from $\bar{a}$. 
From Eqn.~\ref{ideal_product} (\& Claim~\ref{properites_maximal_ideal}), $I = \prod_{\bar{b} \in \calZ_{\F_p}(I)} I_{\bar{b}} = \bigcap_{\bar{b} \in \calZ_{\F_p}(I)} I_{\bar{b}}$.

This finishes the inductive proof, completely factoring $I$. 	
\end{proof}

Lemma~\ref{lemma-MSI-root-represent} shows that a root of a maximal split ideal represents a set of roots of $f\bmod p^k$ and provides the size of that set.

\begin{proof}[Proof of Lemma~\ref{lemma-MSI-root-represent}]
By definition of a maximal split ideal, for any $\bar{a}=(a_0,\ldots,a_l) \in \calZ_{\F_p}(I)$, $p^k | g(x)$ where $g(x)=f(a_0+pa_1+p^2 a_2+\ldots+p^l a_l + p^{l+1} x)$. 
So, $g(x) = 0 \bmod p^k$ for any $p^{k-l-1}$ choices of $x$. 
For each such fixing of $x$, $a_0+pa_1+p^2 a_2+\ldots+p^l a_l + p^{l+1} x$ is a {\em distinct} root of $f(x) \bmod p^k$. Hence proved. 
\end{proof}

\section{Computation modulo a triangular ideal-- Reduce \& Divide}
\label{appen-reduce-valuation}

For completeness, we show that it is efficient to reduce a polynomial $a(\bar{x}_l)\in \G[\bar{x}_l]$ modulo a triangular ideal $J_l =\langle b_0(\bar{x}_0),b_1(\bar{x}_1),\ldots,b_l(\bar{x}_l) \rangle \subseteq \G[\bar{x}_l]$, where $\G$ is any Galois ring (in particular, $R=\Z/p^k$, or $\F_p$).

Note: $J_l$ need not be a split ideal for $f \bmod p^k$, though the algorithms of this section work for split ideals ($\because$ they are triangular by definition).  

Assumptions: In the generators of the triangular ideal we assume $\deg_{x_i}b_i(\bar{x}_i)\ge2$ (for $0\le i\le l$). Otherwise, we could eliminate variable $x_i$ and work with fewer variables (\& smaller length triangular ideal). Additionally, each $b_i(\bar{x}_i)$ (for $0\le i\le l$) is monic (leading coefficient is $1$ wrt $x_i$), and presented in a {\em reduced} form modulo the prior triangular ideal $J_{i-1}:=\langle b_0(\bar{x}_0),\ldots,b_{i-1}(\bar{x}_{i-1})\rangle\subseteq \G[\bar{x}_{i-1}]$.

Let us first define reduction mod an ideal (assume $\G$ to be the Galois ring $G(p^k,b)$).

\begin{definition}[Reduction by a triangular ideal]
\label{def-reduction}
The {\em reduction} of a multivariate polynomial $a(\bar{x}_l) \in \G[\bar{x}_l]$ by a triangular ideal $J_l=\langle b_0(\bar{x}_0),\ldots,b_l(\bar{x}_l) \rangle \subseteq \G[\bar{x}_l]$ 
is the  unique polynomial $\tilde{a}(\bar{x}_l) \equiv a(\bar{x}_l) \bmod  J_l$,
where $\deg_{x_i}(\tilde{a})< \deg_{x_i}(b_i)$, for all $i\in \{0, \ldots,l\}$.
\end{definition}

{\bf Idea of reduction:} The idea behind the algorithm is inspired from the univariate reduction. 
If $l=0$, then reduction of $a(x_0)$ modulo $b_0(x_0)$ is simply the remainder of the division of $a$ by $b_0$ in the underlying polynomial ring $\G[x_0]$. 
For a larger $l$, the reduction of $a(\bar{x}_l)$ modulo the triangular ideal $J_l=\langle b_0(x_0),\ldots,b_l(\bar{x}_l)\rangle$ 
is the remainder of the division of $a(\bar{x}_l)$ by $b_l(\bar{x}_l)$ in the polynomial ring $(\G[x_0,\ldots,x_{l-1}]/J_{l-1})[x_l]$. The fact that $b_l$ is monic, helps in generalizing `long division'.

\smallskip
\textbf{Input:} An $a(\bar{x}_l)\in \G[\bar{x}_l]$ and a triangular ideal $J_l=\langle b_0(\bar{x}_0),\ldots,b_l(\bar{x}_l) \rangle \subseteq \G[\bar{x}_l]$.

\textbf{Output:} Reduction $\tilde{a}$ of $a\bmod J_l$ as defined above.

\begin{breakablealgorithm}
\caption{Reduce $a(\bar{x}_l)$ modulo $J_l$}
\label{algo-reduction}
\begin{algorithmic}[1]
\Procedure{Reduce}{$a(\bar{x}_l)$, $J_{l}$}

\If{$l=0$}

\parState{[\textbf{Reduce $a(x_0)$ by $b_0(x_0)$}] return remainder of univariate division of $a$ by $b_0$ in $R[x_0]$.}

\EndIf

\parState{ $d_a \leftarrow \deg_{x_l}(a)$ and $d_b \leftarrow \deg_{x_l}(b_l)$.}

\parState{Let $a(\bar{x}_l)=: \Sigma_{i=0}^{d_a} a_i(\bar{x}_{l-1}) x_{l}^{i} $ be the polynomial representation of $a(\bar{x}_l)$ with respect to $x_l$.} 

\parState{Recursively reduce each coefficient $a_i(\bar{x}_{l-1})$ of $a$ $\bmod J_{l-1}$: \\ $\tilde{a}_i(\bar{x}_{l-1}) \leftarrow$ \Call{Reduce}{$a_i(\bar{x}_{l-1})$, $J_{l-1}$}, for all $i \in \{0,\ldots,d_a\}$.}


\While{$d_a \geq d_b$}

\parState{$a(\bar{x}_l) \leftarrow a - \left(a_{d_a}\cdot x_l^{d_a-d_b}\cdot b_l\right)$} 

\parState{Update $d_a \leftarrow \deg_{x_l}(a)$. Update $a_i$'s such that $a(\bar{x}_l)=: \Sigma_{i=0}^{d_a} a_i(\bar{x}_{l-1})\cdot x_{l}^{i}$ .} 

\parState{Call \Call{Reduce}{$a_i(\bar{x}_{l-1})$, $J_{l-1}$} for all $i \in \{0,\ldots,d_a\}$: recursively reduce each coefficient $a_i(\bar{x}_{l-1}) \bmod J_{l-1}$ (like Step $7$).}

\EndWhile

\parState{\Return $a(\bar{x}_l)$.}

\EndProcedure
\end{algorithmic}
\end{breakablealgorithm}




Following lemma shows that reduction modulo a triangular ideal (Algorithm~\ref{algo-reduction}) is efficient.

\begin{lemma}[Reduction]
\label{lemma-reduction}
Given $a(\bar{x}_l)\in \G[\bar{x}_l]$ and $J_l\subseteq \G[\bar{x}_l]$, to reduce $a(\bar{x}_l)$ mod $J_l$, Algorithm~\ref{algo-reduction} takes time
$\poly\left(\prod_{i=0}^{l} \deg_{x_i}(a), \log |\G|, \deg(J_l) \right)$. 

In particular, if each coefficient $a_i(\bar{x}_{l-1})$ of $a(\bar{x}_l)$ (viewed as a polynomial in $x_l$) is in reduced form mod $J_{l-1}$, then reduction takes time $\poly\left(d_a, \log |\G|, \deg(J_l)\right)$, where $d_a=\deg_{x_l}(a)$.
\end{lemma}

\begin{proof}
We prove the lemma by induction on the length $l+1$ of the ideal $J_l$. 

For $l=0$, we have a standard univariate reduction which takes at most $O(\deg(a) \deg(b))$ ring operations in $\G$.
Since addition/multiplication/division in $\G$ take time at most $\tilde{O}(\log |\G|)$ \cite{shoup2009computational}, we get the lemma. 

Assume that the lemma is true for any ideal of length less than $l$.

\smallskip
Coefficients $a_i(\bar{x}_{l-1})$ can be reduced, in time $\poly\left(\prod_{i=0}^{l-1} \deg_{x_{i}}(a), \log |\G|, \deg(J_{l-1})\right)$, mod $J_{l-1}$ using induction hypothesis. We need to make $d_a+1$ such calls; total time is bounded by $\poly\left(\prod_{i=0}^{l} \deg_{x_{i}}(a), \log |\G|, \deg(J_{l-1})\right)$. In the same time we can compute Step 9.

After the update at Step $9$, individual-degrees $\deg_{x_i}(a)$ (for $0\le i<l$) can become at most double the previous degree (safely assuming $2\le \deg_{x_i}(b_i)\leq \deg_{x_i}(a)$). 
By induction hypothesis, each call to reduce $a_i(\bar{x}_{l-1}) \bmod J_{l-1}$ takes time $\poly\left(\prod_{i=0}^{l-1} \deg_{x_{i}}(a), \log |\G|, \deg(J_{l-1})\right)$. 
Algorithm makes at most $d_a$ such calls and the while-loop runs at most $d_a$ times. Hence, the algorithm takes time $\poly\left(\prod_{i=0}^{l} \deg_{x_{i}}(a), \log |\G|, \deg(J_l)\right)$; and we are done.

If coefficients of $a$ are already reduced modulo $J_{l-1}$, then $\deg_{x_i}(a)< \deg_{x_i}(b_i)$ for all $0\le i<l$. Hence, Algorithm~\ref{algo-reduction} takes time $d_a^2\cdot\poly\left(\log |\G|, \deg(J_{l-1})\right)$.
\end{proof}

\begin{lemma}[Division mod triangular ideal]\label{lem-div-mod-I}
Given a triangular ideal $J_l\subseteq \G[\bar{x}_l]$ and a unit $a(\bar{x}_l)\in \G[\bar{x}_l]/J_l$. We can compute $a^{-1} \bmod J_l$, in reduced form, in time
$\poly\left(\prod_{i=0}^{l} \deg_{x_i}(a), \log |\G|, \deg(J_l)\right)$. 
\end{lemma}
\begin{proof}
Let $u(\bar{x}_l)\in \G[\bar{x}_l]/J_l$ be such that $u\cdot a \equiv 1 \bmod J_l$.
We can write $u$ as 
$$\sum_{\stackrel{\bar{e}\ \ge\ \bar{0}}{\forall\ 0\le i\le l,\ e_i \ <\ \deg_{x_i}(b_i)}} u_{\bar{e}}\cdot \bar{x}_l^{\bar{e}} \;.$$
We want to find the unknowns $u_{\bar{e}}$ in $\G$, satisfying $u\cdot a \equiv 1 \bmod J_l$. This gives us a linear system in the unknowns; it has size $\deg(J_l)$. The linear system can be written down, using Algorithm~\ref{algo-reduction}, by reducing the monomial products $\bar{x}_l^{\bar{e}} \cdot \bar{x}_l^{\bar{e}'}$ that appear in the product $u\cdot a$. This takes time $\poly\left(\prod_{i=0}^{l} \deg_{x_i}(a), \log |\G|, \deg(J_l)\right)$.

Since there exists a unique $u$, our linear system is efficiently solvable, by standard linear algebra, in the required time.
\end{proof}

Let us see two direct applications of the reduction Algorithm \ref{algo-reduction} to compute valuation and to compute reduced form of split ideals.

First, we explain how Algorithm~\ref{algo1} (Steps $3$, $18$) computes reduced $f_J$ modulo the lift $\hat{J}$ of the newly computed split ideal $J$, 
when $x$ is replaced by $x_{l+1}+p x$ in the intermediate polynomial $f_I(\bar{x}_l,x)$.

\begin{lemma}[Updating stack with reduced polynomial]
\label{lemma-stack-update}
Let $I\subseteq \F_p[\bar{x}_l]$ be a split ideal and $f_I(\bar{x}_l,x)\in R[\bar{x}_l,x]$ be reduced modulo $\hat{I}$ (the lift of $I$ over $R$). 
Define split ideal $J \subseteq \F_p[\bar{x}_{l+1}]$ as  $J:=I+\langle h_{l+1}(\bar{x}_{l+1})\rangle$, and $\hat{J}$ be the lift of $J$ over $R$. 

Then, in time $\poly\left(\log |R|, \deg_{x}(f_I), \deg(J)\right)$, we can compute a reduced polynomial $f_J$ modulo $\hat{J}$ defined by,
$ f_J(\bar{x}_{l+1},x) \;:=\; f_I(\bar{x}_l, x_{l+1}+p x) \bmod \hat{J} \;.$
\end{lemma}
\begin{proof}
Since $f_I(\bar{x}_l,x)$ is already reduced modulo $\hat{I}$, $\deg_{x_i}(f_I)< \deg_{x_i}(h_i)$. 
Define $D:=\deg_{x}(f_I)$, perform the shift $x \to x_{l+1}+p x$ in $f_I$, and expand $f_I$ using Taylor series,
\[ f_J(\bar{x}_l,x) \;=\; f_I(\bar{x}_l, x_{l+1}+p x) \;=:\; g_0(\bar{x}_{l+1}) + g_1(\bar{x}_{l+1}) (p x) +\ldots + g_{D}(\bar{x}_{l+1}) (px)^{D} \;,\]
where $g_i$ could also be seen as the $i$-th derivative of $f_I(\bar{x}_l,x_{l+1})$ (wrt $x_{l+1}$) divided by $i!$. 
To compute $f_J \bmod \hat{J}$, we call \Call{Reduce}{$g_i$, $\hat{J}$} (for all $i$) to get the reduction of each term mod $\hat{J}$.

To calculate the time complexity of \Call{Reduce}{$g_i$, $\hat{J}$}, note that coefficients of each $g_i$, wrt $x_{l+1}$, is already reduced mod $\hat{I}$. 
Since $J=I+\langle h_{l+1}\rangle$, using Lemma~\ref{lemma-reduction}, time complexity of reducing each $g_i$ by $\hat{J}$ is at most $\poly(\deg_{x_{l+1}}(g_i), \log |R|, \deg(J))$ ($\deg(J)=\deg(\hat{J})$).

Since $\deg_{x_{l+1}}(g_i)\leq \deg_{x}(f_I)$ (for $i\leq D$), total time complexity is $\poly\left(\log |R|, \deg_{x}(f_I), \deg(J)\right)$.
\end{proof}

Next, we explain Step $20$ in Algorithm~\ref{algo1} a bit more.

\begin{lemma}[Ideal factors in reduced form]
\label{lemma-split-reduce}
Consider the tuple $(U:=\{h_0(\bar{x}_0),\ldots,h_l(\bar{x}_l)\}, f_{\langle U\rangle})\in S$ and consider a non-trivial factorization $h_i=: h_{i,1}\ldots h_{i,m}$ for some $h_i \in U$. Wlog each factor $h_{i,j}$ is monic wrt $x_i$.

Then, we can compute the factor-related tuples $(U_j, f_{\langle U_j\rangle})$, for all $j\in [m]$, in time $\poly(\deg(\langle U\rangle)$, $\log |R|, \deg_{x}(f_{\langle U\rangle}))$ 
($f_{\langle U_j\rangle}$ will be in reduced form mod $\langle U_j\rangle$).
\end{lemma}
\begin{proof}


First, we successively reduce $h_{i+t}$ ($1\leq t\leq l-i$) modulo triangular ideal $I_{i+t,j} := \langle h_0,\ldots,h_{i-1}, h_{i,j}, h_{i+1},\ldots,h_{i+t}\rangle$.
Time complexity of each of these steps is bounded by  $\poly(\deg(\langle U\rangle)$, $\log |R|)$ (Lemma~\ref{lemma-reduction}). 
This ensures that the degree of $h_{i+t}$ in a variable $x_s$ ($s < i+t$) is less than the individual-degree of the $s$-th generator of ideal $\langle U_j \rangle$.  

Then, $f_{\langle U_j\rangle}$ can be calculated by reducing each $\deg_{x}(f_{\langle U\rangle})+1$ coefficients of $f_{\langle U_j\rangle}$ (wrt $x$) by the lifted triangular ideal $\hat{I}_{l,j} = \hat{U}_j$. 
By Lemma~\ref{lemma-reduction}, this takes time $\poly( \prod_{i=0}^{l}\deg_{x_i}(f_{\langle U\rangle}), \deg_{x}(f_{\langle U\rangle}), \log |R|$, $\deg(\langle U\rangle))$.  
Since coefficients (wrt $x$) of $f_{\langle U\rangle}$ were already reduced modulo $\langle U\rangle$, $\prod_{i=0}^l \deg_{x_i}(f_{\langle U\rangle}) \leq \deg(\langle U\rangle)$. 

So, the computation time is bounded by $\poly\left(\deg(\langle U\rangle), \log |R|, \deg_{x}(f_{\langle U\rangle})\right)$.
\end{proof}

%
%
%

\section{Computation modulo a triangular ideal--- Zerodivisor test \& GCD}
\label{appen-GCD-zerodiv}

\Call{Test-Zero-Div}{$a(\bar{x}_l)$, $I_{l}$}, for a triangular ideal $I_l=:\langle h_0,\ldots,h_l \rangle$, either reports that $a(\bar{x}_l)$ is not a zerodivisor modulo $I_l$, or 
{\em returns a non-trivial factorization} of a generator $h_i=: h_{i,1} \cdots h_{i,m}$ (into monic, wrt $x_i$, factors mod prior ideal). In this section we assume $\F$ to be a finite field.

\smallskip
\textbf{Idea:} In the quotient ring $\F[\bar{x}_l]/\langle I_l \rangle$, a monic (wrt $x_i$) polynomial $a(\bar{x}_i)$ is a zerodivisor iff it contains a factor of 
$h_i(\bar{x}_i)$--- generator of triangular ideal $I_l$ with variables $\{x_0,\ldots,x_i\}$. 
So, firstly the algorithm checks if the given polynomial $a(\bar{x}_l)$ is monic (recursively, from variables $x_{l-1}$ to $x_0$). If it fails, it factors some generator $h_i$ for $i<l$. 
After making $a(\bar{x}_l)$ monic, we take gcd of $a$ with $h_l$--- if it finds non-trivial gcd it factors $h_l$, else $a(\bar{x}_l)$ is not a zerodivisor.

\begin{breakablealgorithm}
 \caption{Zerodivisor test of $a(\bar{x}_l)$ modulo $I_l$}
 \label{algo-gcd}
\begin{algorithmic}[1]
\Procedure{Test-Zero-Div}{$a(\bar{x}_l)$,$I_{l}$}

\If{$l=0$}
\parState{[\textbf{Take univariate GCD}] $gcd \leftarrow \gcd(a(x_0),h_0(x_0)) $.}
\If{$gcd$ is non-trivial}
\parState{Factorize $h_0(x_0)=: gcd \cdot\frac{h_0}{gcd} $; \Return $(True, gcd\cdot \frac{h_0}{gcd})$.}
\Else 
\parState{\Return $(False)$.}
\EndIf
\EndIf
\parState{Let the leading coefficient of $a(\bar{x}_l)$ wrt $x_l$ be $\tilde{a}(\bar{x}_{l-1})$.}
\parState{Call \Call{Test-Zero-Div}{$\tilde{a}(\bar{x}_{l-1})$, $I_{l-1}$}.}
\If{The test returned $True$}
\parState{\Return the result of the test including the factorization of a generator $h_i(\bar{x}_i)$.}
\EndIf

\noindent[Now, we will take gcd of $a$ and $h_l$ using iterated division method (Euclid's method).]

\parState{Define $b(\bar{x}_l) \leftarrow h_l(\bar{x}_l)$.}
\While{$b(\bar{x}_l)\not = 0$}
\parState{Let $\tilde{b}(\bar{x}_{l-1})$ be the leading coefficient of $b(\bar{x}_l)$ wrt $x_l$.}
\If{\Call{Test-Zero-Div}{$\tilde{b}(\bar{x}_{l-1})$, $I_{l-1}$} = $True$}
\parState{\Return result of \Call{Test-Zero-Div}{$\tilde{b}(\bar{x}_{l-1})$,$I_{l-1}$}, factorization of a generator $h_i(\bar{x}_i)$.}
\EndIf
\parState{Let $c(\bar{x}_l) \leftarrow$ \Call{Reduce}{$a(\bar{x}_l)$, $I_{l-1}+\langle b(\bar{x}_l)/\tilde{b} \rangle$} (same as taking remainder of $a(\bar{x}_l)$ when divided by the monic polynomial $b(\bar{x}_l)/\tilde{b}$ modulo $I_{l-1}$).}
\parState{$a(\bar{x}_l) \leftarrow b(\bar{x}_l)/\tilde{b}$, $b(\bar{x}_l) \leftarrow c(\bar{x}_l)$. [Invariant: $\deg_{x_l}(b)$ has fallen.]}
\EndWhile

\noindent[Gcd of original $a(\bar{x}_l)$ and $h_l(\bar{x}_l)$ mod $I_{l}$ is stored in $a(\bar{x}_l)$.]

\If{gcd $a(\bar{x}_l)$ is non-trivial}
\parState{\Return ($True$, a non-trivial factorization of $h_l(\bar{x}_l)$).}
\Else
\parState{\Return $(False)$. [$a(\bar{x}_l)$ is not a zerodivisor.]}
\EndIf

\EndProcedure
\end{algorithmic}
\end{breakablealgorithm}


\begin{lemma}[Efficiency of testing zerodivisors]
\label{lemma-zero-div}
Assuming, coefficients of $a(\bar{x}_l)$ wrt $x_l$ are in reduced form modulo $I_{l-1}$, Algorithm~\ref{algo-gcd} takes time $\poly(\deg_{x_l}(a), \log |\F|, \deg(I_{l}))$.
\end{lemma}
\begin{proof}
We apply induction on the length $l+1$ of ideal $I_l$. 

For $l=0$, it runs univariate gcd and takes time $\poly(\deg(a), \deg(h_0), \log |\F|)$ \cite{shoup2009computational}.
 
Assume lemma statement holds true for ideals of length $l$.

\smallskip
By induction, checking $\tilde{a}(\bar{x}_{l-1})$ is a zerodivisor mod $I_{l-1}$, takes $\poly(\deg_{x_{l-1}}(\tilde{a}), \log |\F|, \deg(I_{l-1}))$ time.

To compute gcd of $a$ and $h_l$, Euclidean gcd algorithm will run at most $\deg_{x_l}(a) + \deg_{x_l}(h_l)$ while-loops. 
From induction hypothesis, and Lemmas \ref{lemma-reduction}-\ref{lem-div-mod-I}, each loop takes at most $\poly(\deg_{x_{l}}(a), \log |\F|$, $\deg(I_{l}))$ time. So, we are done.
\end{proof}

\Call{GCD}{$a(\bar{x}_{l},x)$, $b(\bar{x}_{l},x)$, $I_{l}$} computes gcd of two polynomials $a(\bar{x}_{l},x)$ and $b(\bar{x}_{l},x)$ modulo a triangular ideal 
$I_l=\langle h_0(x_0),\ldots,h_l(\bar{x}_l) \rangle$ resp.~$False$. It computes the {\em monic gcd} resp.~returns a {\em non-trivial factorization} of some $h_i$. 

\begin{breakablealgorithm}
 \caption{GCD computation modulo $I_{l}$}
 \label{algo-gcd}
\begin{algorithmic}[1]
\Procedure{GCD}{$a(\bar{x}_{l},x)$, $b(\bar{x}_{l},x)$, $I_{l}$}
\parState{Let $\tilde{b}(\bar{x}_{l})$ be the leading coefficient of $b$ with respect to $x$.}

\If{\Call{Test-Zero-Div}{$\tilde{b}(\bar{x}_{l})$, $I_{l}$} = $True$}

\parState{\Return $False$, \Call{Test-Zero-Div}{$\tilde{b}(\bar{x}_{l})$, $I_{l}$} factors some generator $h_i(\bar{x}_i)$.}

\EndIf

\parState{Let $c(\bar{x}_l,x) \leftarrow$\Call{Reduce}{$a$, $I_{l}+\langle b/\tilde{b}\rangle$}.}

\If{$c=0$}

\parState{\Return $b/\tilde{b}$.}

\Else

\parState{\Return \Call{GCD}{$b(\bar{x}_l,x)$, $c(\bar{x}_l,x)$, $I_{l}$}.}

\EndIf

\EndProcedure
\end{algorithmic}
\end{breakablealgorithm}

\begin{lemma}[Multivariate GCD]
\label{lemma-gcd}
Algorithm~\ref{algo-gcd} either factors a generator $h_i$ (\& outputs $False$), or computes a monic polynomial $g(\bar{x}_l,x)\in \F[\bar{x}_{l},x]$, such that, 
$g$ divides $a, b$ modulo $I_{l}$. Moreover, $g=u a + v b \bmod I_{l}$, for some $u(\bar{x}_l,x), v(\bar{x}_l,x) \in \F[\bar{x}_l,x]$. 

If $a$ and $b$ are in reduced form mod $I_{l}$, then it takes time $\poly\left(\deg_{x}(a), \deg_{x}(b), \log |\F|, \deg(I_{l})\right)$.
\end{lemma}
\begin{proof}
Algorithm \ref{algo-gcd} is just an implementation of multivariate Euclidean gcd algorithm over the coefficient ring  $\F_p[\bar{x}_{l}]/I_{l} =:R'$. 
If the algorithm outputs $g(\bar{x}_l,x) \in R'[x_l]$ then, by standard Euclidean gcd arguments (using recursion), there exists $u(\bar{x}_l,x), v(\bar{x}_l,x) \in R'[x]$, such that, 
$u a+v b = g$, and $g$ divides both $a$ and $b$ modulo $I_{l}$.

The algorithm works fine if in each step it was able to work with a monic divisor.
Otherwise, it gets stuck at a `division' step, implying that the divisor's leading-coefficient is a zerodivisor, factoring some generator of $I_{l}$. 
	
For time complexity, each recursive step makes one call each to \Call{Test-Zero-Div}{}, \Call{Reduce}{}, and division procedures. 
They take time $\poly\left(\deg_{x}(a), \deg_{x}(b), \log |\F|, \deg(I_{l})\right)$ ($\because$ coefficients of $a$ and $b$ are in reduced form mod $I_{l}$, and use Lemmas \ref{lemma-reduction}, \ref{lem-div-mod-I} \& \ref{lemma-zero-div}).
Since number of recursive steps are bounded by $\deg_{x}(a)+\deg_{x}(b)$, total time is bounded by $\poly\left(\deg_{x}(a), \deg_{x}(b), \log |\F|, \deg(I_{l})\right)$.
\end{proof}

\end{document}